\documentclass{article}

\usepackage{amsfonts, amssymb,amsthm, amsmath, rotating, latexsym}
\usepackage{graphicx,mathtools}

\usepackage{bm}
\usepackage{multirow}
\usepackage{booktabs}
\usepackage{epstopdf}
\usepackage[]{geometry}
\usepackage{mathtools}
\usepackage{bigstrut}
\usepackage{mathrsfs}
\usepackage{setspace}
\usepackage{amsfonts}
\usepackage{amssymb}
\usepackage{amsmath}
\usepackage{enumitem}
\usepackage{latexsym}
\usepackage{multirow}
\usepackage{booktabs}
\usepackage{bm}
\usepackage{multirow}
\usepackage{booktabs}
\usepackage{epstopdf}
\usepackage[caption=false]{subfig}
\usepackage[colorlinks,citecolor=blue,urlcolor=blue,linkcolor=red,naturalnames=true]{hyperref}
\usepackage[hyphenbreaks]{breakurl}
\usepackage[authoryear]{natbib}

\usepackage[countmax]{subfloat}
\newtheorem{thm}{Theorem}
\newtheorem{lemma}{Lemma}

\newtheorem{remark}{Remark}

\newcommand{\be}{\begin{equation}}
\newcommand{\ee}{\end{equation}}
\newcommand{\ba}{\begin{eqnarray}}	
\newcommand{\ea}{\end{eqnarray}}
\newcommand{\bee}{\begin{equation*}}
\newcommand{\eee}{\end{equation*}}
\newcommand{\baa}{\begin{eqnarray*}}
\newcommand{\eaa}{\end{eqnarray*}}

\newcommand{\essinf}{\operatornamewithlimits{ess\,inf}}
\DeclareMathOperator*{\argmin}{arg\,min}
\DeclareMathOperator{\sgn}{sgn}

\newcommand{\R}{\mathbb{R}}
\newcommand{\F}{\mathbb{F}}
\newcommand{\X}{\mathbb{X}}
\newcommand{\D}{\mathcal{D}}

\begin{document}
\title{Estimation of a Two-component Mixture Model with Applications to Multiple Testing}
\author{Rohit Kumar Patra and Bodhisattva Sen \\ Columbia University, USA}
\date{}
\maketitle
\begin{abstract}
We consider a two-component mixture model with one known component. We develop methods for estimating the mixing proportion and the unknown distribution nonparametrically, given i.i.d.~data from the mixture model, using ideas from shape restricted function estimation. We establish the consistency of our estimators. We find the rate of convergence and asymptotic limit of the estimator for the mixing proportion. Completely automated distribution-free honest finite sample lower confidence bounds are developed for the mixing proportion. Connection to the problem of multiple testing is discussed. The identifiability of the model, and the estimation of the density of the unknown distribution are also addressed. We compare the proposed estimators, which are easily implementable, with some of the existing procedures through simulation studies and analyse two data sets, one arising from an application in astronomy and the other from a microarray experiment.
\end{abstract}

\section{Introduction}
Consider a mixture model with two components, i.e.,
\begin{equation}
\label{eq:MixMod}
	F(x) = \alpha F_s(x) + (1 - \alpha) F_b(x),
\end{equation}
where the cumulative distribution function (CDF) $F_b$ is known, but the mixing proportion $\alpha \in [0,1]$ and the CDF $F_s$ ($\ne F_b$) are unknown. Given a random sample from $F$, we wish to (nonparametrically) estimate $F_s$ and the parameter $\alpha$.

This model appears in many contexts. In multiple testing problems (microarray analysis, neuroimaging) the $p$-values, obtained from the numerous (independent) hypotheses tests, are uniformly distributed on [0,1], under $H_0$, while their distribution associated with $H_1$ is unknown; see e.g., \cite{EfronLargeScaleInf10} and \cite{RobinEtAl07}. Translated to the setting of (\ref{eq:MixMod}), $F_b$ is the uniform distribution and the goal is to estimate the proportion of false null hypotheses $\alpha$ and the distribution of the $p$-values under the alternative. In addition, a reliable estimator of $\alpha$ is important when we want to assess or control multiple error rates, such as the false discovery rate of~\cite{BenjHoch95}.

In contamination problems, the distribution $F_b$, for which reasonable assumptions can be made, may be contaminated by an arbitrary distribution $F_s$, yielding a sample drawn from $F$ as in (\ref{eq:MixMod}); see e.g., \cite{McPeel00}. For example, in astronomy, such situations arise quite often: when observing some variable(s) of interest (e.g., metallicity, radial velocity) of stars in a distant galaxy, foreground stars from the Milky Way, in the field of view, contaminate the sample; the galaxy (``signal'') stars can be difficult to distinguish from the foreground stars as we can only observe the stereographic projections and not the three dimensional position of the stars (see \cite{WalkerEtAl09}). Known physical models for the foreground stars help us constrain $F_b$, and the focus is on estimating the distribution of the variable for the signal stars, i.e., $F_s$. We discuss such an application in more detail in Section~\ref{sec:astro}. Such problems also arise in High Energy physics where often the signature of new physics is evidence of a significant-looking peak at some position on top of a rather smooth background distribution; see e.g., \cite{Lyons08}.


Most of the previous work on this problem assume some constraint on the form of the unknown distribution $F_s$, e.g., it is commonly assumed that the distributions belong to certain parametric models, which lead to techniques based on maximum likelihood (see e.g., \cite{Cohen67} and \cite{Lindsay83}), minimum chi-square (see e.g., \cite{Day69}), method of moments (see e.g., \cite{LindsayBasak93}), and moment generating functions (see e.g., \cite{QuandtRamsey78}). \cite{Bordes06} assume that both the components belong to an unknown symmetric location-shift family. \cite{Jin08} and \cite{CaiJin10} use empirical characteristic functions to estimate $F_s$ under a semiparametric normal mixture model. In multiple testing, this problem has been addressed by various authors and different estimators and confidence bounds for $\alpha$ have been proposed in the literature under certain assumptions on $F_s$ and its density, see e.g., \cite{Storey02}, \cite{GenoWass04}, \cite{MeinRice06}, \cite{MeinBuhl05}, \cite{CelisseRobin10} and \cite{LangaasEtAl05}. For the sake of brevity, we do not discuss the above references here but come back to this application in Section~\ref{sec:MultTest}. 

In this paper we provide a methodology to estimate $\alpha$ and $F_s$ (nonparametrically), without assuming any constraint on the form of $F_s$.  The main contributions of our paper can be summarised in the following.
\begin{itemize}
\item  We investigate the identifiability of~\eqref{eq:MixMod} in complete generality. 

\item  When $F$ is a continuous CDF, we develop an honest finite sample lower confidence bound for the mixing proportion $\alpha$. We believe that this is the first attempt to construct a distribution-free lower confidence bound for $\alpha$ that is also tuning parameter-free. 

\item Two different estimators of $\alpha$ are proposed and studied. We derive the rate of convergence and asymptotic limit for one of the proposed estimators.

\item A nonparametric estimator of $F_s$ using ideas from shape restricted function estimation is proposed and its consistency is proved. Further, if $F_s$ has a non-increasing density $f_s$, we can also consistently estimate $f_s$.

\end{itemize}

The paper is organised as follows. In Section~\ref{sec:Mod&Iden} we address the identifiability of the model given in~\eqref{eq:MixMod}. In Section~\ref{sec:Est} we propose an estimator of $\alpha$ and investigate its theoretical properties, including its consistency, rate of convergence and asymptotic limit. In Section~\ref{sec:Lowrbnd} we develop a completely automated distribution-free honest finite sample lower confidence bound for $\alpha$. As the performance of the estimator proposed in Section~\ref{sec:Est} depends on the choice of a tuning parameter, in Section~\ref{sec:Choose_c_n} we study a tuning parameter-free heuristic estimator of $\alpha$.  We discuss the estimation of $F_s$ and its density $f_s$ in Section~\ref{sec:FandDensity}. Connection to the multiple testing problem is developed in Section~\ref{sec:MultTest}. In Section~\ref{sec:Simul} we compare the finite sample performance of our procedures, including a plug-in and cross-validated choice of the tuning parameter for the estimator proposed in Section~\ref{sec:Est}, with other methods available in the literature through simulation studies, and provide a clear recommendation to the practitioner. Two real data examples, one arising in astronomy and the other from a microarray experiment, are analysed in Section~\ref{sec:RealData}. Appendix~\ref{sec:proofs} gives the proofs of the results in the paper. 

\section{The model and identifiability}
\label{sec:Mod&Iden}
\subsection{When $\alpha$ is known}
Suppose that we observe an i.i.d.~sample $X_1, X_2, \ldots, X_n$ from $F$ as in (\ref{eq:MixMod}). If $\alpha \in (0,1]$ were known, a naive estimator of $F_s$ would be
\be
 \hat{F}_{s,n}^{\alpha} = \frac{\mathbb{F}_n-(1-\alpha)F_b}{\alpha},
 \label{eq:naive}
\ee
where $\mathbb{F}_n$ is the empirical CDF of the observed sample, i.e., $\mathbb{F}_n(x)  = \sum_{i=1}^n \mathbf{1}\{X_i \le x\}/n$. Although this estimator is consistent, it does not satisfy the basic requirements of a CDF: $\hat{F}_{s,n}^{\alpha}$ need not be non-decreasing or lie between 0 and 1. This naive estimator can be improved by imposing the known shape constraint of monotonicity. This can be accomplished by minimising
\begin{equation}\label{eq:L2Dist}
\int \{W(x) - \hat{F}_{s,n}^{\alpha}(x)\}^2 \ d\mathbb{F}_n(x) \equiv \frac{1}{n} \sum_{i=1}^n \{W(X_i) - \hat{F}_{s,n}^{\alpha}(X_i)\}^2
\end{equation}
over all CDFs $W$. Let $\check{F}_{s,n}^{\alpha}$ be a CDF that minimises (\ref{eq:L2Dist}). The above optimisation problem is the same as minimising $\| \bm{\theta} - \mathbf{V}\|^2
$ over $\bm{\theta} = (\theta_1, \ldots, \theta_n) \in \Theta_{inc}$ where
\be\label{eq:Theta_Inc}
\Theta_{inc} = \{\bm{\theta} \in \mathbb{R}^n:  0 \le \theta_1 \le \theta_2 \le \ldots \le \theta_n \le 1\}, \nonumber
\ee
$\mathbf{V} = (V_1,V_2,\ldots, V_n)$, $V_i := \hat F_{s,n}^{\alpha}(X_{(i)})$, $i=1,2,\ldots,n$, $X_{(i)}$ being the $i$-th order statistic of the sample, and $\|\cdot \|$ denotes the usual Euclidean norm in $\mathbb{R}^n$. The estimator $\hat{\bm{\theta}}$ is uniquely defined by the projection theorem (see e.g., Proposition 2.2.1 on page 88 of \cite{Bertsekas03}); it is the Euclidean projection of $\mathbf{V}$ on the closed convex set $\Theta_{inc} \subset \mathbb{R}^n$. $\hat{\bm{\theta}}$ is related to $\check F_{s,n}^{\alpha}$ via $\check F_{s,n}^{\alpha}(X_{(i)}) = \hat \theta_i$, and can be easily computed using the pool-adjacent-violators algorithm (PAVA); see Section 1.2 of \cite{RWD88}. Thus, $\check{F}_{s,n}^{\alpha}$ is uniquely defined at the data points $X_i$, for all $i=1,\ldots,n$, and can be defined on the entire real line by extending it to a piece-wise constant right continuous function with possible jumps only at the data points. The following result, derived easily from Chapter 1 of \cite{RWD88}, characterises $\check{F}_{s,n}^{\alpha}$.
\begin{lemma}\label{lemma:Fcheck}
Let $\tilde F_{s,n}^\alpha$ be the isotonic regression (see e.g., page 4 of \cite{RWD88}) of the set of points $\{\hat{F}_{s,n}^{\alpha}(X_{(i)})\}_{i=1}^n$. Then $\tilde F_{s,n}^\alpha$ is characterised as the right-hand slope of the greatest convex minorant of the set of points $\{i/n, \sum_{j=0}^i \hat{F}_{s,n}^{\alpha}(X_{(j)})\}_{i=0}^n$. The restriction of $\tilde{F}_{s,n}^{\alpha}$ to $[0,1]$, i.e., $\check{F}_{s,n}^{\alpha} = \min \{\max \{\tilde{F}_{s,n}^{\alpha}, 0\}, 1\},$ minimises (\ref{eq:L2Dist}) over all CDFs.
\end{lemma}
Isotonic regression and the PAVA are very well studied in the statistical literature with many text-book length treatments; see e.g., \cite{RWD88} and \cite{BarlowEtAl72}. If skillfully implemented, PAVA has a computational complexity of $O(n)$ (see \cite{Grotzinger84}).

\subsection{Identifiability of $F_s$} \label{sec:Ident}
When $\alpha$ is unknown, the problem is considerably harder; in fact, it is non-identifiable. If (\ref{eq:MixMod}) holds for some $F_b$ and $\alpha$ then the mixture model can be re-written as $$ F=(\alpha+\gamma) \left( \frac{\alpha}{\alpha+\gamma} F_s +\frac{\gamma}{\alpha+\gamma} F_b \right) +(1 - \alpha - \gamma)F_b,$$ for $0\leq \gamma \leq 1 - \alpha$, and the term $(\alpha F_s + \gamma F_b)/({\alpha+\gamma})$ can be thought of as the nonparametric component. A trivial solution occurs when we take $\alpha + \gamma = 1$, in which case (\ref{eq:L2Dist}) is minimised when $W = \mathbb{F}_n$. Hence, $\alpha$ is not uniquely defined. To handle the identifiability issue, we redefine the mixing proportion as
\be
 \alpha_0 := \inf \left\{ \gamma \in (0,1]: [F-(1-\gamma)F_b]/{\gamma} \mbox{  is a CDF} \right\}.
\label{eq:alpha.est}
\ee
Intuitively, this definition makes sure that the ``signal'' distribution $F_s$ does not include any contribution from the known ``background'' $F_b$. 

In this paper we consider the estimation of $\alpha_0$ as defined in (\ref{eq:alpha.est}). Identifiability of mixture models has been discussed in many papers, but generally with parametric assumptions on the model. \cite{GenoWass04} discuss identifiability when $F_b$ is the uniform distribution and $F$ has a density. \cite{HunterEtAl07} and \cite{Bordes06} discuss identifiability for location shift mixtures of symmetric distributions. Most authors try to find conditions for the identifiability of their model, while we go a step further and quantify the non-identifiability by calculating $\alpha_0$ and investigating the difference between $\alpha$ and $\alpha_0$. In fact, most of our results are valid even when~\eqref{eq:MixMod} is non-identifiable. 

Suppose that we start with a fixed $F_s, F_b$ and $\alpha$ satisfying~\eqref{eq:MixMod}. As seen from the above discussion we can only hope to estimate $\alpha_0$, which, from its definition in (\ref{eq:alpha.est}), is smaller than $\alpha$, i.e., $\alpha_0 \le  \alpha$. A natural question that arises now is: under what condition(s) can we guarantee that the problem is {\it identifiable}, i.e., $\alpha_0 =  \alpha$? The following lemma gives the connection between $\alpha$ and $\alpha_0$.
\begin{lemma} \label{lemma:Non-Identifiability}
Let $F$ be as in \eqref{eq:MixMod} and $\alpha_0$ as defined in \eqref{eq:alpha.est}. Then 
 \begin{equation} \label{eq:AltDefAlpha}
\alpha_0=\alpha -\sup \left\{ 0 \leq \epsilon \leq 1  : \alpha F_s-\epsilon F_b \mbox{  is a sub-CDF} \right\},
\end{equation}
where sub-CDF is a non-decreasing right-continuous function taking values between 0 and 1. In particular, $\alpha_0<\alpha$ if and only if  there exists $\epsilon \in (0,1)$ such that $\alpha F_s-\epsilon F_b$ is a sub-CDF. Furthermore, $\alpha_0=0$ if and only if $F=F_b.$
\end{lemma}	

In the following we separately identify $\alpha_0$ for any distribution, be it continuous or discrete or a mixture of the two, with a series of lemmas proved in Appendix \ref{sec:Ident_Cont}. By an application of  the Lebesgue decomposition theorem in conjunction with the Jordan decomposition theorem (see page 142, Chapter V, Section $3a^*$ of \cite{Feller2Old}), we have that any CDF $G$ can be uniquely represented as a weighted sum of a piecewise constant CDF $G^{(d)},$ an absolutely continuous CDF $G^{(a)},$ and a continuous but singular CDF $G^{(s)},$ i.e., $G = \eta_1 G^{(a)}+\eta_2 G^{(d)}+\eta_3 G^{(s)}$, where $\eta_i \ge 0$, for $i = 1,2,3$, and $ \eta_1 + \eta_2 + \eta_3 =1$. However, from a practical point of view, we can assume $\eta_3=0,$ since singular functions almost never occur in practice; see e.g., \cite{Parzen60}. Hence, we may assume 
\be \label{eq:genCDF} G = \eta G^{(a)}+(1-\eta)G^{(d)},\ee
 where $(1-\eta)$ is the sum total of all the point masses of $G$. Let $d(G)$ denote the set of all jump discontinuities of $G$, i.e., $d(G) = \{x \in \mathbb{R}: G(x) - G(x-) >0\}$. Let us define $J_G: d(G) \rightarrow [0,1]$ to be a function defined only on the jump points of $G$ such that $J_G(x) = G(x)-G(x-)$ for all $x \in d(G)$. The following result addresses the identifiability issue when both $F_s$ and $F_b$ are discrete CDFs. 


\begin{lemma}\label{lemma:Identifiability for discrete}
Let $F_s$ and $F_b$ be discrete CDFs. If $d(F_b) \not \subset d(F_s),$ then $\alpha_0=\alpha$, i.e.,~\eqref{eq:MixMod} is identifiable.  If $d(F_b)\subset d(F_s),$  then $ \alpha_0= \alpha \left \{1- \inf_{x\in d(F_b)} J_{F_s}(x)/{J_{F_b}(x)}\right\}.$ Thus, $\alpha_0=\alpha$ if and only if $\inf_{x\in d(F_b)} {J_{F_s}(x)}/{J_{F_b}(x)}=0.$
\end{lemma}
Next, let us assume that both $F_s$ and $F_b$ are absolutely continuous CDFs. 
\begin{lemma}\label{lemma:Identifiability for absolutely continuous}
Suppose that $F_s$ and $F_b$ are absolutely continuous, i.e., they have densities $f_s$ and $f_b$, respectively. Then $$\alpha_0=  \alpha \left\{1- \essinf \frac{f_s}{f_b} \right\},$$ where, for any function $g$, $\essinf g = \sup \{a \in \mathbb{R}: \mathfrak{m}(\{x: g(x) <a\})=0\}$, $\mathfrak{m}$ being the Lebesgue measure. As a consequence, $\alpha_0 <  \alpha$ if and only if there exists $c > 0$ such that $f_s\ge c f_b$, almost everywhere w.r.t.~$\mathfrak{m}$.
\end{lemma}
The above lemma states that if there does not exist any $c>0$ for which $f_s(x) \ge c f_b(x)$, for almost every $x$, then $\alpha_0 = \alpha$ and we can estimate the mixing proportion correctly. Note that, in particular, if the support of $F_s$ is strictly contained in that of $F_b$, then the problem is identifiable and we can estimate $\alpha$. 

In Appendix~\ref{sec:Ident_Cont} we apply the above two lemmas to  two discrete (Poisson and binomial) distributions and two absolutely continuous (exponential and normal) distributions to obtain the exact relationship between $\alpha$ and $\alpha_0$. In the following lemma, proved in greater generality in Appendix~\ref{sec:Ident_Cont}, we give conditions under which a general CDF $F$, that can be represented as in \eqref{eq:genCDF}, is identifiable.
\begin{lemma}\label{lemma:Identifiabilty}
Suppose that $
F=\kappa F^{(a)}+ (1-\kappa) F^{(d)},$ where $F^{(a)}$ is an absolutely continuous CDF and $F^{(d)}$ is a piecewise constant CDF, for some $\kappa \in (0,1)$.  Then  \eqref{eq:MixMod} is identifiable, if either $F^{(a)}$ or $F^{(d)}$ are identifiable.
 \end{lemma}
 

\section{Estimation}
\label{sec:Est}
\subsection{Estimation of the mixing proportion $\alpha_0$}\label{sec:EstAlpha}
In this section we consider the estimation of $\alpha_0$ as defined in \eqref{eq:AltDefAlpha}. For the rest of the paper, unless otherwise noted,  we assume 
\[ X_1, X_2, \ldots, X_n\text{ is an i.i.d.~sample from }F\text{ as in \eqref{eq:MixMod}}.\] 

Recall the definitions of $\hat{F}_{s,n}^{\gamma}$  and $\check{F}_{s,n}^{\gamma}$, for $\gamma \in (0,1]$; see (\ref{eq:naive}) and (\ref{eq:L2Dist}). When $\gamma=1$, we have $\hat{F}_{s,n}^{\gamma} =\mathbb{F}_n =\check{F}_{s,n}^{\gamma}$ as $\hat{F}_{s,n}^{\gamma}$ (for $\gamma =1$) is a CDF. Whereas, when $\gamma$ is much smaller than $\alpha_0$ the regularisation of $\hat{F}_{s,n}^{\gamma}$ modifies it, and thus $\hat{F}_{s,n}^{\gamma}$ and $\check{F}_{s,n}^{\gamma}$ are quite different. We would like to compare the naive and isotonised estimators $\hat{F}_{s,n}^{\gamma}$ and $\check{F}_{s,n}^{\gamma}$, respectively, and choose the smallest $\gamma$ for which their distance is still small.
This leads to the following estimator of $\alpha_0$:
\be
\hat{\alpha}_0^{c_n}= \inf \left\{ {\gamma \in (0,1]} :  \gamma  d_n(\hat{F}_{s,n}^{\gamma},\check{F}_{s,n}^\gamma) \le \frac{c_n}{\sqrt{n}} \right\},
\label{eq:EstAlpha}
\ee where $c_n$ is a sequence of constants and $d_n$ stands for the $L_2(\mathbb{F}_n)$ distance, i.e., if $g, h: \mathbb{R} \rightarrow \mathbb{R}$ are two functions, then $d_n^2(g,h)={\int \{g(x)-h(x)\}^2 \ d\mathbb{F}_n(x)}.$ It is easy to see that
\be \label{eq:equiv}
d_n(\mathbb{F}_n, \gamma \check{F}_{s,n}^\gamma +(1-\gamma)F_b)= \gamma d_n(\hat{F}_{s,n}^{\gamma},\check{F}_{s,n}^\gamma).
\ee
For simplicity of notation, using (\ref{eq:equiv}), we define $\gamma  d_n(\hat{F}_{s,n}^{\gamma},\check{F}_{s,n}^\gamma)$ for $\gamma =0$ as 
\be
 \lim_{\gamma \rightarrow 0+} \gamma   d_n(\hat{F}_{s,n}^{\gamma},\check{F}_{s,n}^\gamma) = d_n(\mathbb{F}_n, F_b). \label{eq:DistNull}
 \ee
 This convention is followed in the rest of the paper.

The choice of $c_n$ is important, and in the following sections we address this issue in detail. We derive conditions on $c_n$ that lead to consistent estimators of $\alpha_0$. We will also show that particular (distribution-free) choices of $c_n$ will lead to honest  lower confidence bounds for $\alpha_0$.

Next, we prove a result which implies that, in the multiple testing problem, estimators of $\alpha_0$ do not depend on whether we use $p$-values or $z$-values to perform our analysis. Let  $\Psi:\R \to \R$ be a known continuous non-decreasing function. We define $\Psi^{-1} (y) :=\inf \{t \in \R : y \leq \Psi(t)\},$ and $Y_i:= \Psi^{-1}(X_i).$ It is easy to see that  $Y_1,Y_2, \ldots, Y_n$ is an i.i.d.~sample from $G:=\alpha F_s \circ \Psi + (1-\alpha) F_b \circ\Psi.$ Suppose  now that we work with $Y_1,Y_2, \ldots, Y_n$, instead of $X_1, X_2, \ldots, X_n$, and want to estimate $\alpha$. We can define $\alpha_0^Y$ as in \eqref{eq:alpha.est} but with $\{ G, F_b \circ\Psi\}$ instead of  $\{F, F_b\}$. The following result shows that $\alpha_0$ and its estimators proposed in this paper are invariant under such monotonic transformations.
\begin{thm} \label{thm:Distribution_free}
Let $\mathbb{G}_n$ be the empirical CDF of $Y_1, Y_2,\ldots, Y_n$. Also, let $\hat{G}_{s,n}$ and $ \check{G}^\gamma_{s,n}$ be as defined in \eqref{eq:naive} and~\eqref{eq:L2Dist}, respectively, but with $\{\mathbb{G}_n, F_b \circ\Psi\}$ instead of  $\{\mathbb{F}_n, F_b\}$. Then $\alpha_0=\alpha_0^Y$ and  $\gamma   d_n(\hat{F}_{s,n}^{\gamma},\check{F}_{s,n}^\gamma)= \gamma   d_n(\hat{G}_{s,n}^{\gamma},\check{G}_{s,n}^\gamma)$ for all $\gamma \in (0,1]$.
 \end{thm}
\subsection{Consistency of $\hat{\alpha}_0^{c_n}$}
We start with two elementary results on the behaviour of our criterion function  $\gamma d_n(\check{F}_{s,n}^\gamma,\hat{F}_{s,n}^\gamma)$.

\begin{lemma}\label{lemma1}
For $1 \ge \gamma \ge \alpha_0$,  $\gamma d_n(\check{F}_{s,n}^\gamma,\hat{F}_{s,n}^\gamma) \leq d_n(F,\mathbb{F}_n).$ Thus, 
\be
\gamma d_n(\hat{F}_{s,n}^{\gamma},\check{F}_{s,n}^\gamma) \stackrel{a.s.}{\rightarrow}  \left\{ \begin{array}{l} 0,~~ ~~~\gamma-\alpha_0 \geq 0,~~\\ >0,~~\gamma-\alpha_0 < 0. \end{array} \right.
\label{eq:claim3}
\ee

\end{lemma}

\begin{lemma}\label{lemma2}
The set $A_n := \{\gamma \in [ 0,1 ] : \sqrt{n} \gamma d_n(\hat{F}_{s,n}^{\gamma},\check{F}_{s,n}^\gamma) \le c_n \} $ is convex. Thus, $A_n=[\hat{\alpha}_0^{c_n},1].$
\end{lemma}
The following result shows that for a broad range of choices of $c_n$, our estimation procedure is consistent. 
\begin{thm}\label{thm:ConsAlpha} If  $c_n = o(\sqrt{n})$ and $c_n \rightarrow \infty$, then
$\hat{\alpha}_0^{c_n} \stackrel{P}{\rightarrow} \alpha_0 $.
\end{thm}
A proper choice of $c_n$ is important and crucial for the performance of $\hat{\alpha}_0^{c_n}.$ We suggest doing  cross-validation to find the optimal tuning parameter $c_n$.  In Section~\ref{sec:cv} we detail this approach and illustrate its good finite sample performance through  simulation examples; see Tables \ref{tab:langassIND}-\ref{tab:JinDEP}, Section~\ref{sec:comparison}, and Appendix~\ref{sec:perfor_cont}. However, cross-validation can be computationally expensive. Another useful choice for $c_n$ is to take  $c_n=0.1\log\log n.$ After extensive simulations, we observe that  $c_n=0.1 \log \log n$ has good finite sample performance for estimating $\alpha_0;$ see Section \ref{sec:Simul} and  Appendix \ref{sec:perfor_cont} for more details.
\subsection{Rate of convergence and asymptotic limit} \label{sec:Asymptotics}
We first discuss the case $\alpha_0=0$. In this situation, under minimal assumptions, we show that as the sample size grows, $\hat{\alpha}_0^{c_n}$ exactly equals $\alpha_0$ with probability converging to 1. 
\begin{lemma} \label{lemma:LimDistnull}
When $\alpha_0=0$, if $c_n\rightarrow \infty$ as $n\rightarrow \infty$, then $P(\hat{\alpha}_0^{c_n} =0) \rightarrow 1.$
\end{lemma}

For the rest of this section we assume that $\alpha_0>0$. The following theorem gives the rate of convergence of $\hat{\alpha}_0^{c_n}$. 
\begin{thm} \label{lemma:RateLeftSide}
Let $r_n := \sqrt{n}/c_n$. If $c_n\rightarrow \infty$ and $c_n = o(n^{1/4})$ as $n\rightarrow \infty$, then $r_n(\hat{\alpha}_0^{c_n}-\alpha_0) = O_P(1).$
\end{thm}
The proof of the above result is involved and we give the details in Appendix~\ref{sec:proof_lemma:RateLeftSide}.
\begin{remark}
\cite{GenoWass04} show that the estimators of $\alpha_0$ proposed by \cite{HengartnerStark95} and \cite{Swanepoel99} have convergence rates of $(n/\log n)^{1/3}$ and $n^{2/5}/(\log n)^\delta$, for $\delta >0$, respectively. Morover, both results require smoothness assumptions on $F$  -- \cite{HengartnerStark95} require $F$ to be concave with a density that is Lipschitz of order 1, while \cite{Swanepoel99} requires even stronger smoothness conditions on the density.  \cite{NguyenMatias13} prove that when the density of $F_s^{\alpha_0}$ vanishes at a set of points of measure zero and satisfies certain regularity assumptions, then any $\sqrt{n}$-consistent estimator of $\alpha_0$ will not have finite variance in the limit (if such an estimator exists). 
\end{remark}
We can take $r_n = \sqrt{n}/c_n$ arbitrarily close to $\sqrt{n}$ by choosing $c_n$ that increases to infinity very slowly. If we take $c_n = \log \log n$, we get an estimator that has a rate of convergence $\sqrt{n}/\log \log n$. In fact, as the next result shows, $r_n(\hat{\alpha}_0^{c_n} -\alpha_0)$ converges to a degenerate limit. In Section \ref{sec:PerfEst}, we analyse the effect of $c_n$ on the finite sample  performance of $\hat{\alpha}_0^{c_n}$ for estimating $\alpha_0$ through simulations and advocate a proper choice of the tuning parameter $c_n.$
\begin{thm} \label{lemma:LimDist}
When $\alpha_0>0$, if $r_n \rightarrow \infty$, $c_n= o(n^{1/4})$ and $c_n\rightarrow \infty$, as $n \rightarrow \infty$, then
\bee 
r_n(\hat{\alpha}_0^{c_n} -\alpha_0) \stackrel{P}{\rightarrow} c, \label{eq:LimDist}
\eee
where $c <0$ is a constant that depends on $\alpha_0$, $F$ and $F_b$.
\end{thm}

\section{Lower confidence bound for $\alpha_0$} \label{sec:Lowrbnd}
The asymptotic limit of the estimator $\hat{\alpha}_0^{c_n}$ discussed in Section~\ref{sec:Est} depends on unknown parameters (e.g., $\alpha_0, F$) in a complicated fashion and is of little practical use. Our goal in this sub-section is to construct a finite sample (honest) lower confidence bound $\hat \alpha_L$ with the property
\be
P(\alpha_0 \ge \hat \alpha_L) \ge 1 - \beta,
\label{eq:BoundAlpha}
\ee
for a specified confidence level $(1-\beta)$ ($0 < \beta <1$), that is valid for any $n$ and is tuning parameter free. Such a lower bound would allow one to assert, with a specified level of confidence, that the proportion of ``signal'' is at least $\hat \alpha_L$.

It can also be used to test the hypothesis that there is no ``signal'' at level $\beta$ by rejecting when $\hat \alpha_L >0$. The problem of no ``signal' is known as the homogeneity problem in the statistical literature. It is easy to show that $\alpha_0 = 0$ if and only if $F = F_b$. Thus, the hypothesis of no ``signal'' or homogeneity can be addressed by testing whether $\alpha_0 = 0$ or not. There has been a considerable amount of work on the homogeneity problem, but most of the papers make parametric model assumptions. \cite{Lindsay95} is an authoritative monograph on the homogeneity problem but the components are assumed to be from a known exponential family. \cite{Walther01} and \cite{Walther02} discuss the homogeneity  problem under the assumption that the densities are log-concave. \cite{DonohoJin04} and~\cite{CaiJin10} discuss the problem of detecting sparse heterogeneous mixtures under parametric settings using the `higher criticism' statistic; see Appendix~\ref{sec:Donhojin} for more details.

It will be seen that our approach will lead to an exact lower confidence bound when $\alpha_0 = 0$, i.e., $P(\hat \alpha_L = 0) = 1 - \beta$. The methods of \cite{GenoWass04} and \cite{MeinRice06} usually yield conservative lower bounds.

\begin{thm}\label{thm:BoundAlpha}
Let $H_n$ be the CDF of $\sqrt{n} d_n(\mathbb{F}_n,F)$. Let $\hat \alpha_L$ be defined as in \eqref{eq:EstAlpha} with $c_n= H_n^{-1}(1-\beta)$. Then \eqref{eq:BoundAlpha} holds. Furthermore if $\alpha_0=0$, then $P(\hat{\alpha}_L=0)=1-\beta$, i.e., it is an exact lower bound.
\end{thm}
The proof of the  above theorem can be found in Appendix~\ref{sec:LwrBndAsym_proof}. Note that $H_n$ is distribution-free (i.e., it does not depend on $F_s$ and $F_b$) when $F$ is a continuous CDF and can be readily approximated by Monte Carlo simulations using a sample of uniforms. For moderately large $n$ (e.g., $n \geq 500$) the distribution $H_n$ can be very well approximated by that of the Cram\'{e}r-von Mises statistic, defined as $$\sqrt{n} d(\mathbb{F}_n,F) := \sqrt{\int n \{\mathbb{F}_n(x) - F(x)\}^2 dF(x)}.$$ Letting $G_n$ be the CDF of $\sqrt{n} d(\mathbb{F}_n,F)$, we have the following result.

\begin{thm}\label{thm:ConvDist} $ \sup_{x \in \mathbb{R}} |H_n(x) - G_n(x)| \rightarrow 0 \mbox{ as } n \rightarrow \infty.$
\end{thm}

Hence in practice, for moderately large $n$, we can take $c_n$ to be the $(1-\beta)$-quantile of $G_n$ or its asymptotic limit, which are readily available (e.g., see \cite{Ander52}). When $F$ is a continuous CDF, the asymptotic $95\%$ quantile of $G_n$ is $0.6792$, and is used in our data analysis. Note that 
$$P(\alpha_0 \ge \hat \alpha_L)= P(\sqrt{n} \alpha_0 d_n (\hat{F}_{s,n}^{\alpha_0}, \check{F}_{s,n}^{\alpha_0}) \geq H_n^{-1}(1-\beta)).$$ The following theorem gives the explicit asymptotic limit of $P(\alpha_0 \ge \hat \alpha_L)$ but it is not useful for practical purposes as it involves the unknown $F_s^{\alpha_0}$ and $F$.

\begin{thm} \label{lemma:LwrBndAsym}
Assume that $\alpha_0 >0$. Then $\sqrt{n} \alpha_0 d_n (\hat{F}_{s,n}^{\alpha_0}, \check{F}_{s,n}^{\alpha_0}) \stackrel{d}{\rightarrow} U,$ where  $U$ is a random variable whose distribution depends only  on $\alpha_0,F,$ and $F_b.$  
\end{thm}
The proof of the above theorem and the explicit from of $U$ can be found in Appendix~\ref{sec:proofs}.  The proof of Theorem \ref{thm:ConvDist} and a detailed discussion on the performance of the lower confidence bound for detecting heterogeneity in the {\it moderately sparse} signal regime considered in \cite{DonohoJin04} can be found in Appendix~\ref{sec:Donhojin}.
%
%
\section{A heuristic estimator of $\alpha_0$}
\label{sec:Choose_c_n}
In simulations, we observe that the finite sample performance of \eqref{eq:EstAlpha} is affected by the choice of $c_n$ (for an extensive simulation study on this see Section~\ref{sec:PerfEst}). This motivates us to  propose a method to estimate $\alpha_0$ that is completely automated and has good finite sample performance. We start with a lemma that describes the shape of our criterion function, and will motivate our procedure.
\begin{lemma}\label{lemma:DecCritFn}
$\gamma  d_n(\hat{F}_{s,n}^{\gamma},\check{F}_{s,n}^\gamma) $ is a non-increasing convex function of $\gamma$ in $(0,1)$.
\end{lemma}

Writing $$\hat{F}_{s,n}^{\gamma} = \frac{\mathbb{F}_n - F}{\gamma} + \left\{ \frac{\alpha_0}{\gamma} F_s^{\alpha_0} + \left(1 - \frac{\alpha_0}{\gamma} \right) F_b\right\},$$ we see that for $\gamma \ge \alpha_0$, the second term in the right hand side is a CDF. Thus, for $\gamma \ge \alpha_0$, $\hat{F}_{s,n}^{\gamma}$ is very close to a CDF as $\mathbb{F}_n -F =O_P(n^{-1/2})$, and hence $\check{F}_{s,n}^{\gamma}$ should also be close to $\hat{F}_{s,n}^{\gamma}$. Whereas, for $\gamma < \alpha_0$, $\hat{F}_{s,n}^{\gamma}$ is not close to a CDF, and thus the distance $\gamma d_n(\hat{F}_{s,n}^{\gamma},\check{F}_{s,n}^{\gamma})$ is appreciably large. Therefore, at $\alpha_0$, we have a ``regime'' change:  $\gamma d_n(\hat{F}_{s,n}^{\gamma},\check{F}_{s,n}^{\gamma})$ should have a slowly decreasing segment to the right of $\alpha_0$ and a steeply non-increasing segment to the left of $\alpha_0$.  
\begin{figure}[h!]

  \captionsetup[subfigure]{labelformat=empty}
\subfloat[]{
\includegraphics[width=2.4in, height=1.6in]{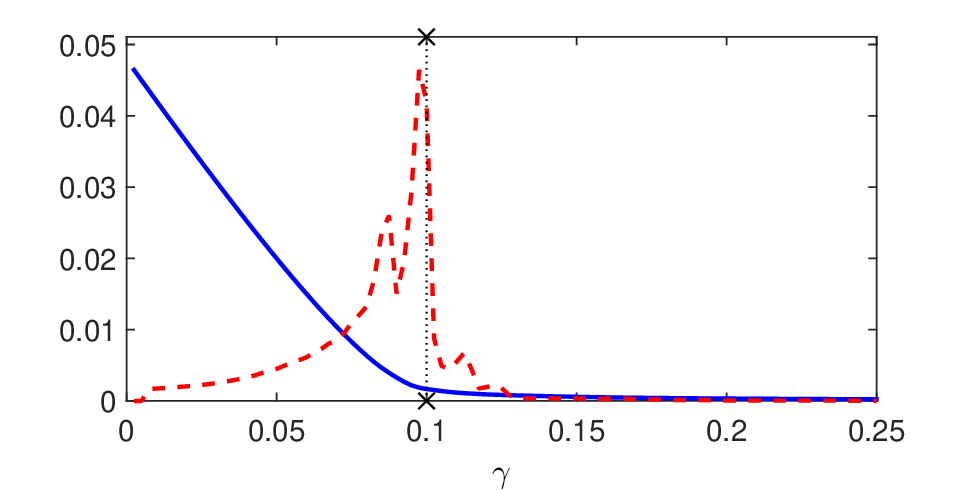}}
\subfloat[]{
\includegraphics[width=2.4in, height=1.6in]{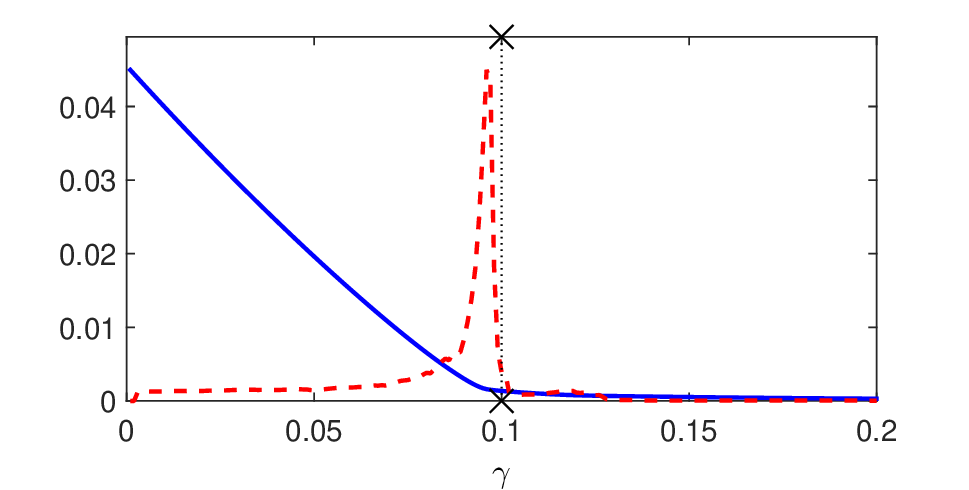}}\\[-6ex]
\caption{Plots of $\gamma d_n(\hat{F}_{s,n}^{\gamma},\check{F}_{s,n}^{\gamma})$ (in solid blue) overlaid with its (scaled) second derivative (in dashed red) for $\alpha_0=0.1$ and $n=5000$. Left panel: setting I; right panel: setting II.}
\label{fig:PlotCritFunc} 
\end{figure}
Fig.~\ref{fig:PlotCritFunc} shows two typical such plots of the function $\gamma  d_n(\hat{F}_{s,n}^{\gamma},\check{F}_{s,n}^{\gamma})$, where the left panel corresponds to a mixture of $N(2,1)$ with $N(0,1)$ (setting I) and in the right panel  we have a mixture of Beta(1,10) and Uniform$(0,1)$ (setting II). We will use these two settings to illustrate our methodology in the rest of this section and also in Section~\ref{sec:LowrbndSim}.

Using the above heuristics, we can see that the ``elbow'' of the function should provide a good estimate of $\alpha_0$; it is the point that has the maximum curvature, i.e., the point where the second derivative is maximal. We denote this estimator by $\tilde \alpha_0$. Notice that both the estimators $\tilde \alpha_0$ and $\hat{\alpha}_0^{c_n}$ are derived from $\gamma d_n( \hat{F}_{s,n}^\gamma ,\check{F}_{s,n}^\gamma)$, as a function of $\gamma$, albeit they look at two different aspects of the function.

In the above plots we have used numerical methods to approximate the second derivative of $\gamma d_n(\hat{F}_{s,n}^{\gamma},\check{F}_{s,n}^{\gamma})$ (using the method of double differencing). We advocate plotting the function $\gamma d_n(\hat{F}_{s,n}^{\gamma},\check{F}_{s,n}^{\gamma})$ as $\gamma$ varies between 0 and 1. In most cases, plots similar to Fig.~\ref{fig:PlotCritFunc} would immediately convey to the practitioner the most appropriate choice of $\tilde \alpha_0$. In some cases though, there can be multiple peaks in the second derivative, in which case some discretion on the part of the practitioner might be required. It must be noted that the idea of finding the point where the second derivative is large to detect an ``elbow'' or ``knee'' of a function is not uncommon; see e.g., \cite{SC04}. However, in Section \ref{sec:comparison} and Appendix \ref{sec:perfor_cont}, we show some simulation examples where $\tilde{\alpha}_0$ fails to consistently estimate the ``elbow'' of $\gamma d_n(\hat{F}_{s,n}^{\gamma},\check{F}_{s,n}^{\gamma}).$ 

\section{Estimation of the distribution function and its density} \label{sec:FandDensity}
\subsection{Estimation of $F_s$}
\label{sec:EstConcaveFs}

Let us assume for the rest of this section that \eqref{eq:MixMod} is identifiable, i.e., $\alpha = \alpha_0,$ and $\alpha_0>0$. Thus $F_s^{\alpha_0} = F_s$. Once we have a consistent estimator $\check \alpha_n$ (which may or may not be $\hat{\alpha}_0^{c_n}$ as discussed in the previous sections) of $\alpha_0$, a natural nonparametric estimator of $F_s$ is $\check F_{s,n}^{\check \alpha_n}$, defined as the minimiser of (\ref{eq:L2Dist}). In the following theorem  we show that, indeed, $\check F_{s,n}^{\check \alpha_n}$ is uniformly consistent for estimating $F_s$. We also derive the rate of convergence of $\check F_{s,n}^{\check \alpha_n}$.

\begin{thm}\label{thm:ConsF_sn}
Suppose that $\check \alpha_n \stackrel{P}{\rightarrow} \alpha_0$. Then, as $n \rightarrow \infty$, $\sup_{x \in \R} | \check F_{s,n}^{\check \alpha_n}(x) - F_s(x) | \stackrel{P}{\rightarrow} 0.$ Furthermore, if  $q_n( \check \alpha_n- \alpha_0) =O_P(1),$ where $q_n = o(\sqrt{n})$, then $\sup_{x \in \mathbb{R}} q_n | \check F_{s,n}^{\check \alpha_n}(x) - F_s(x) | = O_P(1).$ Additionally, for $\hat{\alpha}_0^{c_n}$ as defined in \eqref{eq:EstAlpha}, we have \[ \sup_{x \in \mathbb{R}}    | r_n (\hat{F}_{s,n}^{\hat{\alpha}_0^{c_n}} -F_s)(x)  - Q(x)| \stackrel{P}{\rightarrow} 0 \quad \text{and} \quad  r_n d(\check{F}_{s,n}^{\hat{\alpha}_0^{c_n}} ,F_s)\stackrel{P}{\rightarrow} c\] for a function $Q:\R \to \R$ and a constant $c >0$ depending only on $\alpha_0, F$, and $F_b$. 
 \end{thm}
An immediate consequence of Theorem \ref{thm:ConsF_sn} is that $d_n(\check F_{s,n}^{\check \alpha_n}, \hat F_{s,n}^{\check \alpha_n}) \stackrel{P}{\rightarrow} 0$ as $n \rightarrow \infty$. Left panel of Fig.~\ref{fig:LCMFDensexampleB} shows our estimator $\check F_{s,n}^{\check \alpha_n}$ along with the true $F_s$ for the same data set used in the right panel of Fig.~\ref{fig:PlotCritFunc}.
\begin{figure}[h!]

  \captionsetup[subfigure]{labelformat=empty}
\centering
\subfloat[]{
\label{fig:LCMCDFexample:b} 
\includegraphics[width=2.4in, height=1.3in]{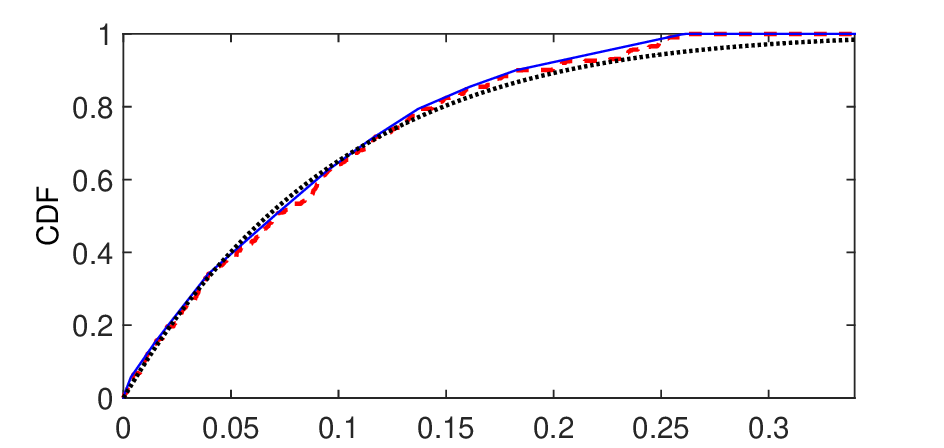}}
\subfloat[]{
\label{fig:PlotDensityexample:b} 
\includegraphics[width=2.4in, height=1.6in]{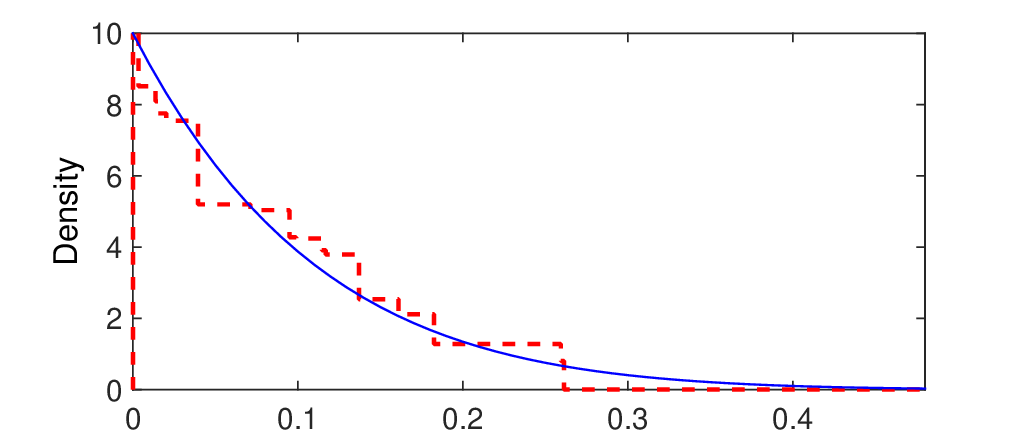}}\\[-5ex]
\caption{Left panel: Plots of $\check F_{s,n}^{\tilde{\alpha}_0}$ (in dashed red), $F_{s,n}^{\dagger}$  (in solid blue) and $F_s$ (in dotted black) for setting II; right panel: plots of $f_{s,n}^{\dagger}$  (in dashed red) and $f_s$ (in solid blue) for setting II. }
\label{fig:LCMFDensexampleB}  

\end{figure}
\subsection{Estimating the density of $F_s$}
Suppose now that $F_s$ has a density $f_s$. Obtaining nonparametric estimators of $f_s$ can be difficult as it requires smoothing and usually involves the choice of tuning parameter(s) (e.g., smoothing bandwidths), and especially so in our set-up.

In this sub-section we describe a tuning parameter free approach to estimating $f_s$, under the additional assumption that $f_s$ is non-increasing. The assumption that $f_s$ is non-increasing, i.e., $F_s$ is concave on its support, is natural in many situations (see Section~\ref{sec:MultTest} for an application in the multiple testing problem) and has been investigated by several authors, including \cite{Grenander56}, \cite{LangaasEtAl05} and \cite{GenoWass04}. Without loss of generality, we assume that $f_s$ is non-increasing on $[0,\infty)$.

For a bounded function $g: [0,\infty) \rightarrow \R$, let us represent the least concave majorant (LCM) of $g$ by $LCM[g]$. Thus, $LCM[g]$ is the smallest concave function that lies above $g$. Define $F_{s,n}^{\dagger} := LCM[\check F_{s,n}^{\check \alpha_n}]$. Note that $F_{s,n}^{\dagger}$ is a valid CDF. We can now estimate $f_s$ by $f_{s,n}^{\dagger}$, where $f_{s,n}^{\dagger}$ is the piece-wise constant function obtained by taking the left derivative of $F_{s,n}^{\dagger}$. In the following result we show that both $F_{s,n}^{\dagger}$ and $f_{s,n}^{\dagger}$ are consistent estimators of their population versions.
\begin{thm}\label{thm:Consden_sn}
Assume that $F_s(0) = 0$ and that $F_s$ is concave on $[0,\infty)$. If $\check \alpha_n \stackrel{P}{\rightarrow} \alpha_0$, then, as $n \rightarrow \infty$,
\be\label{eq:SupConvF_snDagger}
\sup_{x \in \R} | F_{s,n}^{\dagger}(x) - F_s(x) | \stackrel{P}{\rightarrow} 0.
\ee
Further, if for any $x >0$, $f_s(x)$ is continuous at $x$, then, $f_{s,n}^{\dagger}(x) \stackrel{P}{\rightarrow} f_s(x).$
\end{thm}

Computing $F_{s,n}^{\dagger}$ and $f_{s,n}^{\dagger}$ are straightforward, an application of the PAVA gives both the estimators; see e.g., Chapter 1 of \cite{RWD88}. In Fig.~\ref{fig:LCMFDensexampleB}  the left panel shows the LCM $F_{s,n}^{\dagger}$ whereas the right panel shows its derivative $f_{s,n}^{\dagger}$ along with the true density $f_s$ for the same data set used in  the right panel of Fig.~\ref{fig:PlotCritFunc}.

\section{Multiple testing problem}
\label{sec:MultTest}
The problem of estimating the proportion of false null hypotheses $\alpha_0$ is of interest in situations where a large number of hypothesis tests are performed. Recently, various such situations have arisen in applications. One major motivation is in estimating the proportion of genes that are differentially expressed in deoxyribonucleic acid (DNA) microarray experiments. However, estimating the proportion of true null hypotheses is also of interest, for example, in functional magnetic resonance imaging (see \cite{Turkheimer2001}) and source detection in astrophysics (see \cite{Milleretal01}).

Suppose that we wish to test $n$ null hypotheses $H_{01}, H_{02}, \ldots, H_{0n}$ on the basis of a data set $\X$. Let $H_i$ denote the (unobservable) binary variable that is $0$ if $H_{0i}$ is true, and 1 otherwise, $i=1,\ldots, n$. We want a decision rule $\D$ that will produce a decision of ``null'' or ``non-null'' for each of the $n$ cases. 
In their seminal work, \cite{BenjHoch95} argued that an important quantity to control is the false discovery rate (FDR) and proposed a procedure with the property FDR $\le \beta (1-\alpha_0)$, where $\beta$ is the user-defined level of the FDR procedure. When $\alpha_0$ is significantly bigger than $0$ an estimate of $\alpha_0$ can be used to yield a procedure with FDR approximately equal to $\beta$ and thus will result in an increased power. This is essentially the idea of the adapted control of FDR (see \cite{BenjHoch00}). See \cite{Storey02}, \cite{Black04}, \cite{LangaasEtAl05}, \cite{BenjaminiEtAl06}, and  \cite{DonohoJin04} for a discussion on the importance of efficient estimation of $\alpha_0$ and some proposed estimators.

Our method can be directly used to yield an estimator of $\alpha_0$ that does not require the specification of any tuning parameter, as discussed in Section~\ref{sec:Choose_c_n}. We can also obtain a completely nonparametric estimator of $F_s$, the distribution of the $p$-values arising from the alternative hypotheses. Suppose that $F_b$ has a density $f_b$ and $F_s$ has a density $f_s$. To keep the following discussion more general, we allow $f_b$ to be any known density, although in most multiple testing applications we will take $f_b$ to be Uniform$(0,1)$. The {\it local false discovery rate} (LFDR) is defined as the function $l: (0,1) \rightarrow [0,\infty)$, where
\be\label{eq:lfdr}
l(x) = P (H_i = 0| X_i = x) = \frac{(1-\alpha_0) f_b(x)}{f(x)}, \nonumber
\ee
and $f(x) = \alpha_0 f_s(x) +  (1 - \alpha_0) f_b(x)$ is the density of the observed $p$-values. The estimation of the LFDR $l$ is important because it gives the probability that a particular null hypothesis is true given the observed $p$-value for the test. The LFDR method can help us get easily interpretable thresholding methods for reporting the ``interesting'' cases (e.g., $l(x) \le 0.20$). Obtaining good estimates of $l$ can be tricky as it involves the estimation of an unknown density, usually requiring smoothing techniques; see Section 5 of \cite{EfronLargeScaleInf10} for a discussion on estimation and interpretation of $l$. From the discussion in Section~\ref{sec:EstConcaveFs}, under the additional assumption that $f_s$ is non-increasing, we have a natural tuning parameter free estimator $\hat l$ of the LFDR:
\be\label{eq:EstLFDR}
	\hat l(x)  = \frac{(1- \check \alpha_n) f_b(x)}{ \check \alpha_n f_{s,n}^{\dagger}(x) + (1 - \check \alpha_n) f_b(x) }, \qquad \mbox{ for } x \in (0,1).\nonumber
\ee
 The assumption that $f_s$ is non-increasing, i.e., $F_s$ is concave, is quite natural -- when the alternative hypothesis is true the $p$-value is generally small -- and has been investigated by several authors, including \cite{GenoWass04} and \cite{LangaasEtAl05}.


\section{Simulation}
\label{sec:Simul}

To investigate the finite sample performance of the estimators developed in this paper, we carry out several simulation experiments. We also compare the performance of these estimators with existing methods. The R language (\cite{Rlang}) codes used to implement our procedures are available at http://stat.columbia.edu/$\sim$rohit/research.html.

\subsection{Lower bounds for $\alpha_0$}\label{sec:LowrbndSim}
%
%
%
%
%
%
%
\begin{table}
\caption{\label{tab:n1000and5000}Coverage probabilities of nominal 95\% lower confidence bounds for the three methods  when $n = 1000$ and $n=5000$. }
\centering
\begin{tabular}{*{13}{c}}
\toprule
 &\multicolumn{6}{c}{$n=1000$}  &\multicolumn{6}{c}{$n=5000$}  \\
\cmidrule(r){2-7} \cmidrule(l){8-13}
&\multicolumn{3}{c}{Setting I} &\multicolumn{3}{c}{Setting II} &\multicolumn{3}{c}{Setting I} &\multicolumn{3}{c}{Setting II} \\
\cmidrule(r){2-4} \cmidrule(rl){5-7}\cmidrule(rl){8-10}\cmidrule(l){11-13}
$\alpha$ & $\hat{\alpha}_L$ & $\hat{\alpha}_L^{GW}$ & $\hat{\alpha}_L^{MR}$ &  $\hat{\alpha}_L$ & $\hat{\alpha}_L^{GW}$ & $\hat{\alpha}_L^{MR}$ & $\hat{\alpha}_L$ & $\hat{\alpha}_L^{GW}$ & $\hat{\alpha}_L^{MR}$ &  $\hat{\alpha}_L$ & $\hat{\alpha}_L^{GW}$ & $\hat{\alpha}_L^{MR}$   \\
\midrule
0 & 0.95 & 0.98 & 0.93 &0.95 & 0.98 & 0.93 & 0.95	& 0.97 & 0.93 & 0.95 & 0.97 & 0.93\\
0.01 & 0.97 & 0.98 & 0.99 & 0.97 & 0.97 & 0.99 & 0.98 & 0.98 & 0.99 & 0.98 & 0.98 & 0.99\\
0.03 & 0.98 & 0.98 & 0.99 & 0.98 & 0.98 & 0.99 & 0.98 & 0.98 & 0.99 & 0.98 & 0.98 & 0.99 \\
0.05 & 0.98 & 0.98 & 0.99 & 0.98 & 0.98 & 0.99 & 0.99 & 0.99 & 0.99 & 0.98 & 0.98 & 0.99\\
0.10 & 0.99 & 0.99 & 1.00 & 0.99 & 0.98 & 0.99 & 0.99 & 0.99 & 1.00 & 0.99 & 0.98 & 0.99\\ 
\bottomrule
\end{tabular}

\end{table}	
Although there has been some work on estimation of $\alpha_0$ in the multiple testing setting, \cite{MeinRice06} and \cite{GenoWass04} are the only papers we found that discuss methodology for constructing lower confidence bounds for $\alpha_0$. These procedures are connected and the methods in \cite{MeinRice06} are extensions of those proposed in \cite{GenoWass04}. The lower bounds proposed in both the papers approximately satisfy~\eqref{eq:BoundAlpha} and have the form $\sup_{t \in (0,1)} (\mathbb{F}_n(t)-t-\eta_{n, \beta} \delta(t))/(1-t),$ where $\eta_{n, \beta}$ is a \textit{bounding sequence} for the \textit{bounding function} $\delta(t)$ at level $\beta$; see \cite{MeinRice06}. \cite{GenoWass04} use a constant bounding function, $\delta(t)=1$, with $\eta_{n, \beta}=\sqrt{\log (2/\beta)/{2n}}$, whereas \cite{MeinRice06} suggest a class of bounding functions but observe that the \textit{standard deviation-proportional} bounding function $\delta(t) =\sqrt{t(1-t)}$ has optimal properties among a large class of possible bounding functions. We use this bounding function and a bounding sequence suggested by the authors. We denote the lower bound proposed in \cite{MeinRice06} by $\hat{\alpha}_L^{MR}$, the bound in \cite{GenoWass04} by $\hat{\alpha}_L^{GW}$, and the lower bound discussed in Section~\ref{sec:Lowrbnd} by $\hat{\alpha}_L$. To be able to use the methods of \cite{MeinRice06} and \cite{GenoWass04} in setting I, introduced in Section \ref{sec:Choose_c_n}, we transform the data such that $F_b$ is $\text{Uniform}(0,1)$ ; see Section \ref{sec:EstAlpha} for the details.

We take $\alpha \in \{0, 0.01, 0.03, 0.05, 0.10\}$ and compare the performance of the three lower bounds in the two different simulation settings discussed in Section~\ref{sec:Choose_c_n}. For each setting we take the sample size $n$ to be $1000$ and $5000$. We present the estimated coverage probabilities, obtained by averaging over $5000$ independent replications, of the lower bounds for both settings in Table \ref{tab:n1000and5000}. We can immediately see from the table that the bounds are usually quite conservative. However, it is worth pointing out that when $\alpha_0=0$, our method has exact coverage, as discussed in Section~\ref{sec:Lowrbnd}. Also, the fact that our procedure is simple, easy to implement, and completely automated, makes it very attractive.

\subsection{Estimation of $\alpha_0$}\label{sec:PerfEst}

In this sub-section, we illustrate and compare the performance of different estimators of $\alpha_0$ under two sampling scenarios. 
In scenario A, we proceed as in \cite{LangaasEtAl05}. 
Let $\mathbf{X}_j = (X_{1j}, X_{2j}, \ldots ,X_{nj})$, for $j = 1,\ldots, J$, and assume that each $\mathbf{X}_j \sim N(\mu_{n \times 1}, \Sigma_{n \times n})$ and that $\mathbf{X}_1, \mathbf{X}_2, \ldots , \mathbf{X}_J$ are independent. We test $H_{0i}:\mu_i=0$ versus $H_{1i}: \mu_i \ne 0$ for each $i=1,2,\ldots,n$. We set $\mu_i$ to zero for the true null hypotheses, whereas for the false null hypotheses, we draw $\mu_i$ from a symmetric bi-triangular density with parameters $a=\log_2 (1.2)=0.263$ and $b=\log_2(4)=2$; see page 568 of \cite{LangaasEtAl05} for the details. Let $x_{ij}$  denote a realisation of $X_{ij}$ and $\alpha$ be the proportion of false null hypotheses. Let  $\bar{x}_i= \sum_{j=1}^J x_{ij}/J$ and  $ s_i^2= \sum_{j=1}^J (x_{ij}-\bar{x}_i)^2/ (J-1)$. To test $H_{0i}$ versus $H_{1i}$, we calculate a two-sided $p$-value based on a one-sample $t$-test, with $p_i = 2 P(T_{J-1} \geq |\bar{x}_i/ \sqrt{s_i^2/J}|)$, 
 where $T_{J-1}$ is a $t$-distributed random variable with $J-1$ degrees of freedom. 

In scenario B, we generate $n+L$ independent random variables $w_1,w_2, \ldots ,w_{n+L}$ from $N(0,1)$ and set $z_i = \frac{1}{\sqrt{L+1}} \sum_{j=i}^{i+L} w_j$ for $i =1,2, \ldots ,n$. The  dependence structure of the $z_i$'s is determined by $L$. For example, $L=0$ corresponds to the case where the $z_i$'s are i.i.d.~standard normal. Let $X_i= z_i +m_i $, for $i =1,2, \ldots ,n$, where $m_i=0$ under the null, and under the alternative, $|m_i|$ is randomly generated from $\text{Uniform}(m^*, m^*+1)$ and $\sgn(m_i)$, the sign of $m_i$, is randomly generated from $\{-1,1\}$ with equal probabilities. Here $m^*$ is a suitable constant that describes the simulation setting. Let $1-\alpha$ be the proportion of true null hypotheses. 
Scenario B is inspired by the numerical studies in \cite{CaiJin10} and \cite{Jin08}.
 \begin{figure}[h!]
\centering
  \captionsetup[subfigure]{labelformat=empty}
\subfloat[]{
\includegraphics[width=2.6in,height=2.3in]{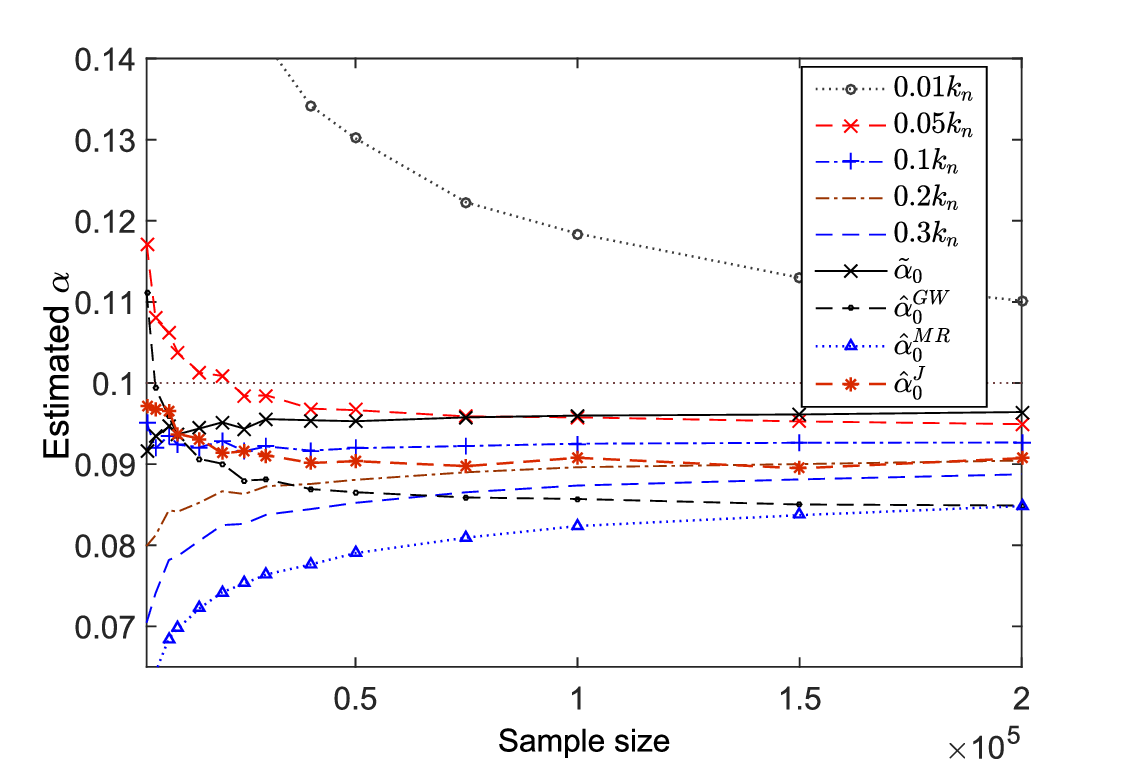}}
\subfloat[]{
\includegraphics[width=2.6in,height=2.3in]{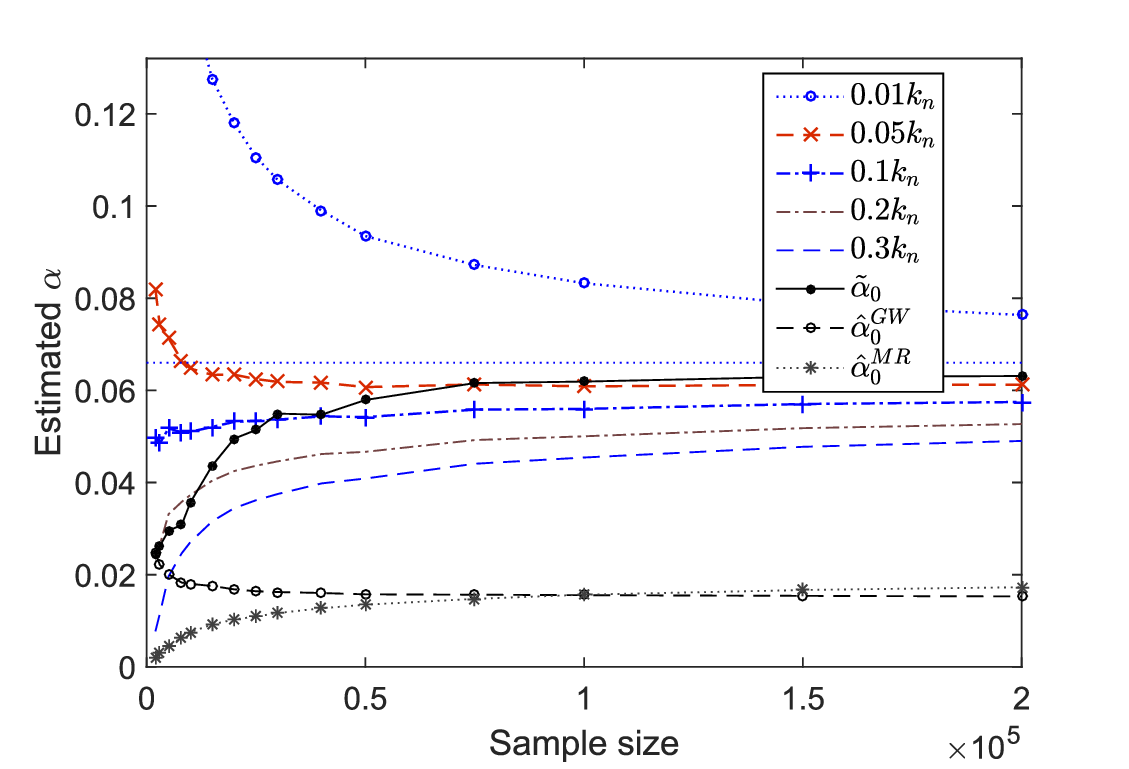}}\\[-5ex]
\caption{ Plots of the means of different estimators of $\alpha_0$, computed  over 500 independent replications, as the sample size increases from $3000$ to $2\times 10^5$; left panel: scenario A with  $\Sigma =I_{n \times n};$ right panel: scenario B with  $L=0$ and $m^*=1$. The horizontal line (in dotted blue) indicates the value of $\alpha_0.$ }
\label{fig:choosingc_n} 
\end{figure} 

We use $\hat{\alpha}_0^{S,B}$ to denote the estimator proposed by \cite{Storey02} when bootstrapping is used to choose the required tuning parameter, and denote by $\hat{\alpha}_0^{S,\lambda}$ the estimator when the value of the tuning parameter is fixed at $\lambda.$ \cite{LangaasEtAl05} proposed an estimator that is tuning parameter free but crucially uses the known shape constraint of a convex and non-increasing $f_s$; we denote it by $\hat{\alpha}_0^{L}$. We evaluate $\hat{\alpha}_0^L$ using the {\tt convest} function in the R library {\tt limma}. We also use the estimator proposed in \cite{MeinRice06} for two bounding functions: $\delta(t)=\sqrt{t(1-t)}$ and $\delta(t)=1$. For its implementation, we must choose a sequence $\{\beta_n\}$ going to zero as $n \rightarrow \infty$. \cite{MeinRice06} did not specify any particular choice of $\{\beta_n\}$ but required the sequence satisfy some conditions. We choose $\beta_n = 0.05/\sqrt{n}$ and denote the estimators by $\hat{\alpha}_0^{MR}$ when $\delta(t)=\sqrt{t(1-t)}$ and by $\hat{\alpha}_0^{GW}$ when $\delta(t)=1$ (see \cite{GenoWass04}). We also compare our results with $\hat{\alpha}_0^E$, the estimator proposed in \cite{Efron07} using the central matching method, computed using  the {\tt locfdr} function in the R library {\tt locfdr}. \cite{Jin08} and \cite{CaiJin10} propose estimators when the model is a mixture of Gaussian distributions; we denote the  estimator proposed in Section 2.2 of \cite{Jin08} by $\hat{\alpha}_0^J$ and in Section 3.1 of \cite{CaiJin10} by $\hat{\alpha}_0^{CJ}.$ Some of the competing methods require $F_b$ to be of a specific form (e.g., standard normal) in which case we transform the observed data suitably.

The estimator $\hat \alpha_0^{c_n}$ depends on the choice of $c_n$ and in the following we investigate a proper choice of $c_n$. We take $\alpha_0=0.1$ and evaluate the performance of  $\hat{\alpha}_0^{\tau \times \log \log n}$ for different values of $\tau$, as $n$ increases, for scenarios A and B. The choice $c_n=\tau \times \log \log n$, for different values of $\tau$, is suggested after extensive simulations. We also include $\tilde{\alpha}_0$, $\hat{\alpha}_0^{GW}$,  $\hat{\alpha}_0^{MR}$, and $\hat{\alpha}_0^J$ in the comparison.   For scenario A, we fix the sample size $n$ at $5000$ and   $\Sigma= I_{n \times n}$. For scenario B, we fix $n =5 \times 10^4,$ $L=0,$ and $m^*=1.$  In Fig.~\ref{fig:choosingc_n}, we illustrate  the effect of $c_n$ on estimation of $\alpha_0$ as $n$ varies from $3000$ to $10^5.$ Recall that $\tilde{\alpha}_0$ denotes the estimator proposed in Section~\ref{sec:Choose_c_n}. For both scenarios, the sample mean of the estimators of $\alpha_0$ proposed in this paper converge to the true $\alpha_0$, as the sample size grows. The methods developed in this paper perform favorably in comparison to $\hat{\alpha}_0^{GW}$,  $\hat{\alpha}_0^{MR}$, and $\hat{\alpha}_0^J.$ Since, the choice of $c_n$ dictates the finite sample performance of $\hat{\alpha}_0^{c_n},$ we propose cross-validation to find an appropriate value of the tuning parameter. 
\begin{table}
\caption{ \label{tab:langassIND} Means$\times 10$  and RMSEs$\times 100$ (in parentheses) of  estimators discussed in Section~\ref{sec:PerfEst} for scenario A with $\Sigma=I_{n \times n},$ $J=10$, $n=5000$, and $k_n= \log  \log n.$} 
\centering
\begin{tabular}{*{11}{c}}
\toprule
    $10\alpha_0 $&$\hat{\alpha}_0^{.1 k_n} $& $\hat{\alpha}_0^{CV} $ & $\tilde{\alpha}_0$ & $\hat{\alpha}_0^{GW}$ & $\hat{\alpha}_0^{MR}$ & $\hat{\alpha}_0^{S,0.5}$ &  $\hat{\alpha}_0^J$ & $\hat{\alpha}_0^{CJ}$  &  $\hat{\alpha}_0^{L}$  &  $\hat{\alpha}_0^ {E}$\\
\midrule 
 0.10  & 0.13  & 0.15  & 0.13  & 0.00  & 0.01  & 0.09  & 0.14  & 0.05  & 0.16  & 0.36 \\
          & (1.00) & (1.79) & (0.83) & (1.00) & (0.88) & (1.41) & (1.50) & (5.32) & (1.20) & (3.70) \\
    0.30  & 0.30  & 0.35  & 0.27  & 0.02  & 0.12  & 0.29  & 0.29  & 0.15  & 0.35  & 0.36 \\
          & (1.02) & (1.87) & (1.01) & (2.80) & (1.84) & (1.41) & (1.83) & (5.46) & (1.26) & (3.96) \\
    0.50  & 0.48  & 0.51  & 0.46  & 0.18  & 0.26  & 0.47  & 0.49  & 0.26  & 0.55  & 0.35 \\
          & (1.09) & (1.9) & (1.12) & (3.29) & (2.46) & (1.49) & (1.91) & (5.73) & (1.34) & (3.80) \\
    1.00  & 0.93  & 0.97  & 0.93  & 0.62  & 0.65  & 0.95  & 0.96  & 0.51  & 1.02  & 0.33 \\
          & (1.35) & (1.86) & (1.32) & (3.88) & (3.57) & (1.51) & (1.94) & (7.16) & (1.36) & (3.73) \\
\bottomrule
\end{tabular}
\end{table}
\begin{table}
	\caption{ \label{tab:JinIND} Means$\times 10$  and RMSEs$\times 100$ (in parentheses) of  estimators discussed in Section~\ref{sec:PerfEst} for scenario B with $L=0,$ $m^*=1,$ $n=5 \times 10^4$, and $k_n= \log \log n.$}
	\centering
	\begin{tabular}{*{11}{c}} \toprule
		$10\alpha_0$ & $\hat{\alpha}_0^{.1 k_n} $& $\hat{\alpha}_0^{CV}$  &$\tilde{\alpha}_0$ & $\hat{\alpha}_0^{GW}$ & $\hat{\alpha}_0^{MR}$ & $\hat{\alpha}_0^{S,B}$ &  $\hat{\alpha}_0^J$ & $\hat{\alpha}_0^{CJ}$  &  $\hat{\alpha}_0^{L}$  &  $\hat{\alpha}_0^ {E}$ \\
		\midrule
		
    0.07  & 0.03  & 0.04  & 0.08  & 0.00  & 0.00  & 0.04  & 0.11  & 0.19  & 0.03  & 0.06 \\
          & (0.44) & (0.67) & (0.28) & (0.66) & (0.66) & (0.65) & (0.96) & (2.96) & (0.38) & (0.77) \\
    0.20  & 0.14  & 0.18  & 0.16  & 0.00  & 0.01  & 0.08  & 0.28  & 0.55  & 0.07  & 0.05 \\
          & (0.73) & (0.79) & (0.62) & (1.98) & (1.89) & (2.25) & (1.33) & (4.41) & (1.26) & (1.28) \\
    0.33  & 0.25  & 0.31  & 0.28  & 0.02  & 0.04  & 0.12  & 0.48  & 0.92  & 0.12  & 0.05 \\
          & (0.89) & (0.85) & (0.95) & (3.15) & (2.91) & (3.83) & (1.77) & (6.48) & (2.14) & (1.90) \\
    0.66  & 0.55  & 0.62  & 0.58  & 0.12  & 0.14  & 0.23  & 0.95  & 1.83  & 0.23  & 0.05 \\
          & (1.21) & (1.00) & (1.48) & (5.38) & (5.25) & (7.73) & (3.04) & (11.98) & (4.34) & (3.84) \\
     \bottomrule
	\end{tabular}
\end{table}
\subsubsection{Cross-validation} \label{sec:cv} 
 In this sub-section, we use $c$ instead of $c_n$ to simplify the notation. In the following we briefly describe our cross-validation procedure. For a $K$-fold cross validation, we randomly partition the data into $K$ sets, say $\D_1,\ldots,\D_K.$ Let $\mathbb{F}_n^{k}$ be the empirical CDF of the data in $\D_k.$ Let $\hat{\alpha}_{0,-k}^{c}$ be the estimator defined in \eqref{eq:EstAlpha} using all data except those in $\D_k$ and tuning parameter $c$. Further, let $\check{F}_{s,n}^{\hat{\alpha}_{0,-k}^{c},-k}$ be the estimator of $F_s$  as defined in Lemma \ref{lemma:Fcheck} using $\hat{\alpha}_{0,-k}^{c}$ and all data except those in $\D_k.$ Define the cross-validated  estimator of $c$ as
 \be \label{eq:cv} c_{cv} := \argmin_{c\in \R} \sum_{k=1}^K \int (\mathbb{F}_n^{k}- \hat{F}^{k})^2 d\mathbb{F}_n^{k}, \ee
where  $\hat{F}^{k}:=\hat{\alpha}_{0,-k}^{c} \check{F}_{s}^{\hat{\alpha}_{0,-k}^{c},-k} +(1-\hat{\alpha}_{0,-k}^{c}) F_b.$   In all simulations in this paper, we use $K=10$ and  denote this estimator by $\hat{\alpha}_0^{CV};$ see Section 7.10 of~\cite{ESL} for a more detailed study of cross-validation and a justification for $K=10$. Fig.~\ref{fig:comparison_plots} illustrates the superior performance of $\hat{\alpha}_0^{CV}$ across different simulation settings; also see Sections~\ref{sec:SimInd} and~\ref{sec:comparison}, and Appendix~\ref{sec:perfor_cont} 
\subsubsection{Performance under independence}
\label{sec:SimInd}
In this sub-section, we take $\alpha \in \{0.01, 0.03, 0.05, 0.10\}$ and compare the performance of the different estimators under the independence setting of scenarios A and B.  In Tables~\ref{tab:langassIND} and \ref{tab:JinIND}, we give the mean and root mean squared error (RMSE)  of the estimators over 5000 independent replications. For scenario A, we fix the sample size $n$ at $5000$ and   $\Sigma= I_{n \times n}$. For scenario B, we fix $n =5 \times 10^4,$ $L=0,$ and $m^*=1.$ By an application of Lemma \ref{lemma:Identifiability for absolutely continuous}, it is easy to see that in scenario A, the model is identifiable (i.e., $\alpha_0 =\alpha$), while in scenario B, $\alpha_0=\alpha \times 0.67$. For scenario A, the sample means of $\hat{\alpha}_0^{CV},$ $\tilde{\alpha}_0,$ $\hat{\alpha}_0^J,$ $\hat{\alpha}_0^L,$ and $\hat{\alpha}_0^{0.1 k_n}$ for $k_n=\log \log n$ are comparable. However, the RMSEs of $\tilde{\alpha}_0$ and  $\hat{\alpha}_0^{0.1 k_n}$ are lower than those of $\hat{\alpha}_0^{CV},$ $\hat{\alpha}_0^J,$ and  $\hat{\alpha}_0^L.$   For scenario B, the sample means of $\tilde{\alpha}_0,$ $\hat{\alpha}_0^{CV},$ and  $\hat{\alpha}_0^{0.1 k_n}$ are comparable.  In scenario B, the performances of $\hat{\alpha}_0^J$ and $\hat{\alpha}_0^{CJ}$ are not comparable to the estimators proposed in this paper, as $\hat{\alpha}_0^J$ and  $\hat{\alpha}_0^{CJ}$ estimate $\alpha,$ while $\tilde{\alpha}_0,$ $\hat{\alpha}_0^{CV},$ and $\hat{\alpha}_0^{c_n}$ estimate $\alpha_0.$ Note that $\hat{\alpha}_0^L$ fails to estimate $\alpha_0$ because the underlying assumption inherent in their estimation procedure, that $f_s$ be non-increasing, does not hold.  In scenario A, $\hat{\alpha}_0^{S,0.5}$ has the best performance among the different values of $\lambda,$ while in scenario B, $\hat{\alpha}_0^{S,\lambda}$ has poor performance for all values of $\lambda \in [0,1].$ Furthermore, $\hat{\alpha}_0^{GW},$ $\hat{\alpha}_0^{MR},$ $\hat{\alpha}_0^{CJ}, \hat{\alpha}_0^{S,B}$ and $\hat{\alpha}_0^ {E}$ perform poorly  in both scenarios for all values of $\alpha_0.$

\subsubsection{Performance under dependence} 
\label{sec:SimDep}
The simulation settings of this sub-section are designed to investigate the effect of dependence on the performance of the estimators. For scenario A, we use the setting of \cite{LangaasEtAl05}. We take  $\Sigma$ to be a block diagonal matrix with block size 100. Within blocks, the diagonal elements (i.e., variances) are set to 1 and the off-diagonal elements (within-block correlations) are set to $\rho=0.5$. Outside of the blocks, all entries are set to 0. Tables \ref{tab:langassDEP} and \ref{tab:JinDEP} show that in both scenarios, none of the methods perform well for small values of $\alpha_0.$ However, in scenario A,  the performances of  $\hat \alpha_0^{0.1 k_n},$ $\tilde{\alpha}_0,$ and $\alpha_0^J$ are comparable, for larger values of $\alpha_0.$ In scenario B,  $\hat \alpha_0^{0.1 k_n}$ performs well for $\alpha_0= 0.033$ and $0.067$.  Observe that, as in the independence setting, $\hat{\alpha}_0^{GW},$ $\hat{\alpha}_0^{MR},$ $\hat{\alpha}_0^{S,B},$  $\hat{\alpha}_0^{CJ},$ and $\hat{\alpha}_0^ {E}$ perform poorly  in both scenarios for all values of $\alpha_0.$
\begin{table}
\caption{ \label{tab:langassDEP} Means$\times 10$  and RMSEs$\times 100$ (in parentheses) of  estimators discussed in Section~\ref{sec:PerfEst} for scenario A with $\Sigma$ as described in Section~\ref{sec:SimDep}, $J=10$, $n=5000$, and $k_n= \log \log n$.}
\begin{tabular}{*{11}{c}} \toprule
    $10\alpha_0$&  $\hat{\alpha}_0^{.1 k_n} $&$\hat{\alpha}_0^{CV} $ & $\tilde{\alpha}_0$ & $\hat{\alpha}_0^{GW}$ & $\hat{\alpha}_0^{MR}$ & $\hat{\alpha}_0^{S,0.5}$ &  $\hat{\alpha}_0^J$ & $\hat{\alpha}_0^{CJ}$  &  $\hat{\alpha}_0^{L}$  &  $\hat{\alpha}_0^ {E}$\\
\midrule    0.10  & 0.46  & 0.42  & 0.33  & 0.07  & 0.06  & 0.28  & 0.22  & 0.07  & 0.32  & 0.37 \\
          & (5.15) & (4.23) & (3.84) & (1.72) & (1.27) & (4.11) & (3.03) & (10.61) & (4.37) & (3.91) \\
    0.30  & 0.52  & 0.53  & 0.41  & 0.14  & 0.17  & 0.65  & 0.34  & 0.15  & 0.49  & 0.39 \\
          & (3.80) & (3.64) & (3.59) & (2.72) & (1.90) & (6.58) & (3.25) & (10.35) & (4.30) & (4.31) \\
    0.50  & 0.66  & 0.76  & 0.54  & 0.26  & 0.31  & 0.54  & 0.49  & 0.25  & 0.66  & 0.37 \\
          & (3.52) & (5.43) & (3.85) & (3.56) & (2.50) & (2.61) & (3.60) & (10.45) & (4.31) & (4.03) \\
    1.00  & 1.06  & 1.13  & 0.97  & 0.68  & 0.69  & 1.15  & 0.97  & 0.53  & 1.11  & 0.36 \\
          & (3.09) & (3.92) & (4.00) & (4.15) & (3.54) & (6.01) & (3.61) & (10.55) & (4.13) & (3.99) \\
         \bottomrule
\end{tabular}

\end{table}
\begin{table}
\caption{ \label{tab:JinDEP} Means$\times 10$  and RMSEs$\times 100$ (in parentheses) of  estimators discussed in Section~\ref{sec:PerfEst} for scenario B with $L=30$, $m^*=1$, $n=5 \times 10^4$, and $k_n= \log \log n$.}

\begin{tabular}{*{11}{c}}
\toprule
     $10\alpha_0$ &  $\hat{\alpha}_0^{.1 k_n} $&  $\hat{\alpha}_0^{CV} $ &$ \tilde{\alpha}_0$ & $\hat{\alpha}_0^{GW}$ & $\hat{\alpha}_0^{MR}$ & $\hat{\alpha}_0^{S,B}$ &  $\hat{\alpha}_0^J$ & $\hat{\alpha}_0^{CJ}$  &  $\hat{\alpha}_0^{L}$  &  $\hat{\alpha}_0^ {E}$\\
\midrule
 
    0.07  & 0.29  & 0.38  & 0.17  & 0.04  & 0.05  & 0.26  & 0.20  & 0.21  & 0.13  & 0.22 \\
          & (2.92) & (3.70) & (1.62) & (1.02) & (1.36) & (3.71) & (2.80) & (9.87) & (1.75) & (2.22) \\
    0.20  & 0.30  & 0.42  & 0.18  & 0.04  & 0.04  & 0.16  & 0.33  & 0.55  & 0.13  & 0.19 \\
          & (1.84) & (2.88) & (1.25) & (1.75) & (1.71) & (2.24) & (3.25) & (10.35) & (1.42) & (2.27) \\
    0.33  & 0.38  & 0.52  & 0.20  & 0.06  & 0.06  & 0.17  & 0.50  & 0.93  & 0.16  & 0.18 \\
          & (1.54) & (2.74) & (1.89) & (2.83) & (2.73) & (3.51) & (3.71) & (11.52) & (2.03) & (2.59) \\
    0.67  & 0.63  & 0.77  & 0.31  & 0.14  & 0.15  & 0.24  & 0.95  & 1.82  & 0.25  & 0.16 \\
          & (1.53) & (2.25) & (4.32) & (5.26) & (5.13) & (7.60) & (4.54) & (15.13) & (4.23) & (4.08) \\
\bottomrule  
\end{tabular}
\end{table}
 \begin{figure}[h!]
\vspace{0cm}
\centering
\includegraphics[width=5.8in,height=4.3in]{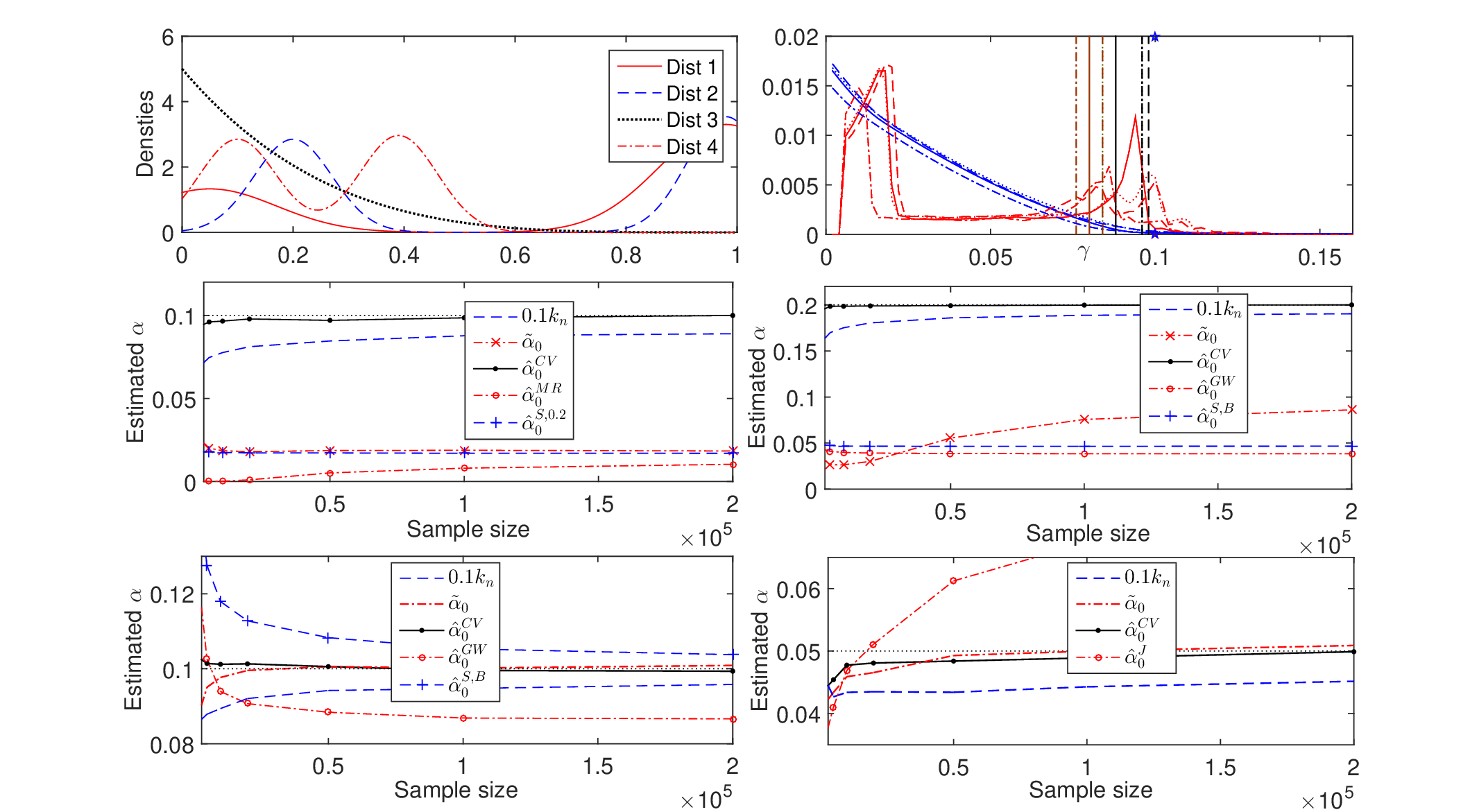}
\caption{Top row left panel: density functions for different choices of $F_s$;  top row right panel: plot of $\gamma d_n(\hat{F}_{s,n}^{\gamma},\check{F}_{s,n}^{\gamma})$ (in blue), the scaled second derivative (in red), $\hat{\alpha}_0^{CV}$ (in black), and  $\hat{\alpha}_0^{0.1 k_n}$ (in brown)  for 5 independent samples of size 5000  corresponding to ``Dist 1"; the blue star denotes $\alpha_0$. The bottom two rows show the means of different competing estimators of $\alpha_0$, computed over 500 independent samples for Dist 1-4 (left-right, top-bottom) as sample size increases from $3000$ to $2\times10^5$; in each figure the dotted black line denotes the true $\alpha_0.$}
\label{fig:comparison_plots} 
\vspace*{0cm}
\end{figure}
\subsubsection{Comparing the performance of $\hat{\alpha}_0^{c_n},$ $\hat{\alpha}_0^{CV},$ and $\tilde{\alpha}_0$} \label{sec:comparison} 
Although the heuristic estimator $\tilde \alpha_0$ performs quite well in most of the simulation settings considered, there exists scenarios where $\tilde \alpha_0$ can fail to consistently estimate $\alpha_0$. To illustrate this we consider four different CDFs $F_s$ and fix $F_b$ to be the uniform distribution on $(0,1)$ (see the top left plot of Fig.~\ref{fig:comparison_plots}) and compare the performance of $\hat{\alpha}_0^{CV}$, $\tilde{\alpha}_0$, $\hat{\alpha}_0^{0.1 k_n}$ with the best performing competing estimators (in each setting). 

We see that $\tilde{\alpha}_0$ may fail to estimate the ``elbow'' of $\gamma d_n(\hat{F}_{s,n}^{\gamma},\check{F}_{s,n}^{\gamma})$, as a function of $\gamma$, when $F_s$ has a multi-modal density (see the middle row of Fig.~\ref{fig:comparison_plots}). Observe that $\hat{\alpha}_0^{CV}$ and $\hat{\alpha}_{0}^{0.1k_n}$ perform favorably compared to all competing estimators and in the two scenarios where $\tilde{\alpha}_0$ fails to consistently estimate ${\alpha}_0$, all our competing estimators also fail. 

The first two toy examples have been carefully constructed to demonstrate situations where the point of maximum curvature ($\tilde{\alpha}_0$) is different from the ``elbow'' of the function; see the top right plot of Fig.~\ref{fig:comparison_plots} (also see Appendix \ref{sec:perfor_cont} for further such examples). 

\begin{figure}[h!]
\captionsetup[subfigure]{labelformat=empty}
\centering
\includegraphics[width=5in, height=3in]{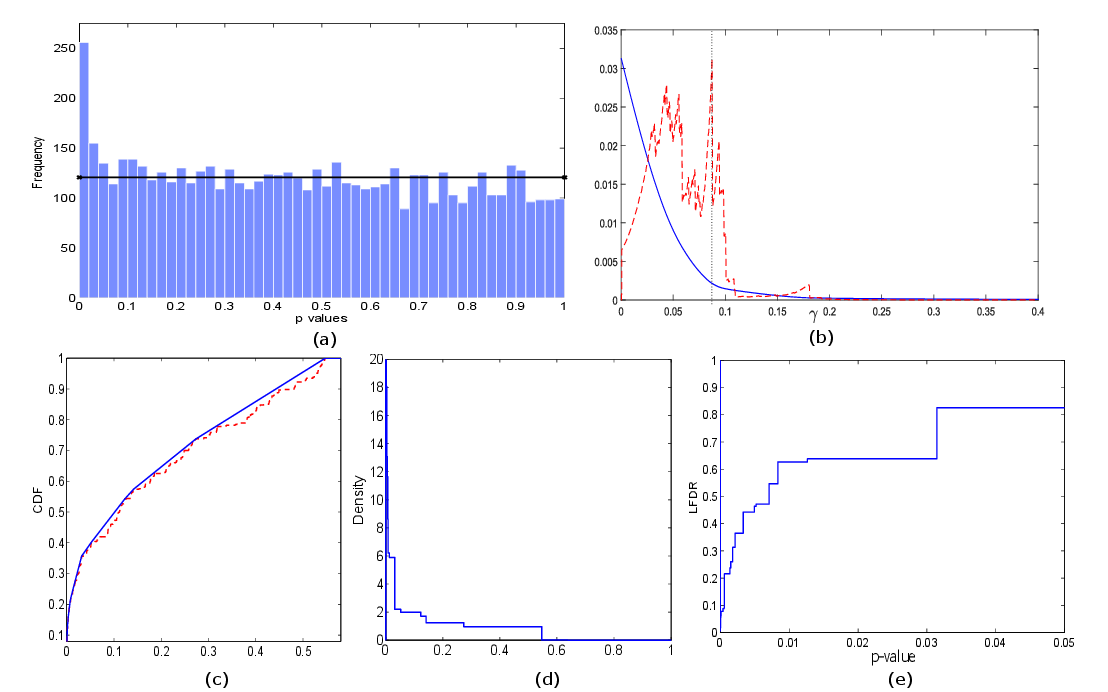}\\ [-2ex]
\caption{ Plots for the prostate data:  (a) Histogram of the $p$-values. The horizontal line (in solid black) indicates the Uniform$(0,1)$ distribution. (b) Plot of $\gamma d_n(\hat{F}_{s,n}^{\gamma},\check{F}_{s,n}^{\gamma})$ (in solid blue) overlaid with its (scaled) second derivative (in dashed red). The vertical line (in dotted black) indicates the point of maximum curvature $\tilde{\alpha}_0 = 0.088$. (c) $\check F_{s,n}^{\tilde{\alpha}_0}$ (in dotted red) and $F_{s,n}^{\dagger}$  (in solid blue); (d) $f_{s,n}^{\dagger}$; (e) estimated LFDR $\hat l$ for $p$-values less than 0.05.}
\label{fig:Prostatedens_dist} 

\end{figure}
\subsubsection{Our recommendation} In this paper we study two estimators for $\alpha_0$. For $\hat{\alpha}_0^{c_n}$, a proper choice of $c_n$ is important for good finite sample performance. We suggest using cross-validation to find the optimal tuning parameter $c_n$. However, cross-validation can be computationally expensive. An attractive alternative in this situation is to use $\tilde{\alpha}_0$, which is easy to implement and has very good finite sample performance in most scenarios, especially with large sample sizes.  We feel that a visual analysis of the plot of $\gamma d_n(\hat{F}_{s,n}^{\gamma},\check{F}_{s,n}^{\gamma})$ can be useful in checking the validity of $\tilde{\alpha}_0$ as an estimator of the ``elbow'', and thus for ${\alpha}_0$. 


\section{Real data analysis}
\label{sec:RealData}		
\subsection{Prostate data}
\label{sec:prostate}

Genetic expression levels for $n= 6033$ genes were obtained for $m = 102$ men, $m_1 = 50$ normal control subjects and $m_2 = 52$ prostate cancer patients. Without going into the biology involved, the principal goal of the study was to discover a small number of  ``interesting'' genes, that is, genes whose expression levels differ between the cancer and control patients. Such genes, once identified, might be further investigated for a causal link to prostate cancer development.
The prostate data is a 6033 $\times$ 102 matrix $\mathbb{X}$ having entries $x_{ij}=\mbox{expression level for gene } i \mbox{ on patient }j$, $i=1,2,\ldots,n$, and $j=1,2, \ldots,m,$ with $j=1,2,\ldots,50$, for the normal controls, and $j=51,52,\ldots,102$, for the cancer patients. Let $\bar{x}_i(1)$ and $\bar{x}_i(2)$ be the averages of $x_{ij}$ for the normal controls and for the cancer patients, respectively, for gene $i$. The two-sample $t$-statistic for testing significance of gene $i$ is $t_i=\{\bar{x}_i(1)-\bar{x}_i(2)\}/s_i,$ where $s_i $ is an estimate of the standard error of  $\bar{x}_i(1)-\bar{x}_i(2)$, i.e.,
$s_i^2 = (1/50 + 1/52) [\sum_{j=1}^{50} \{x_{ij}-\bar{x}_i(1)\}^2 + \sum_{j=51}^{102} \{x_{ij}-\bar{x}_i(2)\}^2 ]/100.$

We work with the $p$-values obtained from the 6033 two-sided $t$-tests instead of the ``$t$-values" as then the distribution under the alternative will have a non-increasing density which we can estimate using the method developed in Section~\ref{sec:EstConcaveFs}.   Note that in our analysis we ignore the dependence of the $p$-values, which is only a moderately risky assumption for the prostate data; see Chapters 2 and 8 of \cite{EfronLargeScaleInf10} for further analysis and justification. 
Fig.~\ref{fig:Prostatedens_dist} show the plots of various quantities of interest, found using the methodology developed in Section~\ref{sec:EstConcaveFs} and Section~\ref{sec:MultTest}, for the prostate data example. The 95\% lower confidence bound $\hat{\alpha}_L$ for this data is found to be $0.05$.
In Table~\ref{tab:realdata}, we display estimates of $\alpha_0$ based on the methods considered in this paper for the prostate data and the Carina data (described below). 
\subsection{Carina data -- an application in astronomy} \label{sec:astro}

In this sub-section we analyse the radial velocity (RV) distribution of stars in Carina, a dwarf spheroidal (dSph) galaxy. The dSph galaxies are low luminosity galaxies that are companions of the Milky Way. The data have been obtained by Magellan and MMT telescopes (see \cite{Mattetal07}) and consist of radial (line of sight) velocity measurements of $n = 1266$ stars from Carina, contaminated with Milky Way stars in the field of view. We would like to understand the distribution of the RV of stars in Carina. For the contaminating stars from the Milky Way in the field of view we assume a non-Gaussian velocity distribution $F_b$ that is known from the Besancon Milky Way model (\cite{RobinEtAl03}), calculated along the line of sight to Carina.

\begin{table}
\caption{ \label{tab:realdata}Estimates of $\alpha_0$ for the two data sets.} 
\centering
\begin{tabular}{*{12}{l}}
\toprule
    Data set &  $\hat{\alpha}_0^{0.1 k_n} $&$\hat{\alpha}_0^{CV}$ & $\tilde{\alpha}_0$ & $\hat{\alpha}_0^{GW}$ & $\hat{\alpha}_0^{MR}$ & $\hat{\alpha}_0^{S,B}$ &  $\hat{\alpha}_0^J$ & $\hat{\alpha}_0^{CJ}$  &  $\hat{\alpha}_0^{L}$  &  $\hat{\alpha}_0^ {E}$ \\
    \midrule
 Prostate      &  0.08 & 0.10 & 0.09   &  0.04 &  0.01   & 0.19 &        0.10 &   0.02 & 0.11 & 0.02 \\
 Carina  &  0.36 &0.35& 0.36 & 0.31 & 0.30 &		0.45	&	0.61 &		1.00 & 0.38	& NA \\
 \bottomrule
\end{tabular}

\end{table}
\begin{figure}[h!]
\centering

\captionsetup[subfigure]{labelformat=empty}
\subfloat[]{
\includegraphics[width=1.8in, height=1.6in]{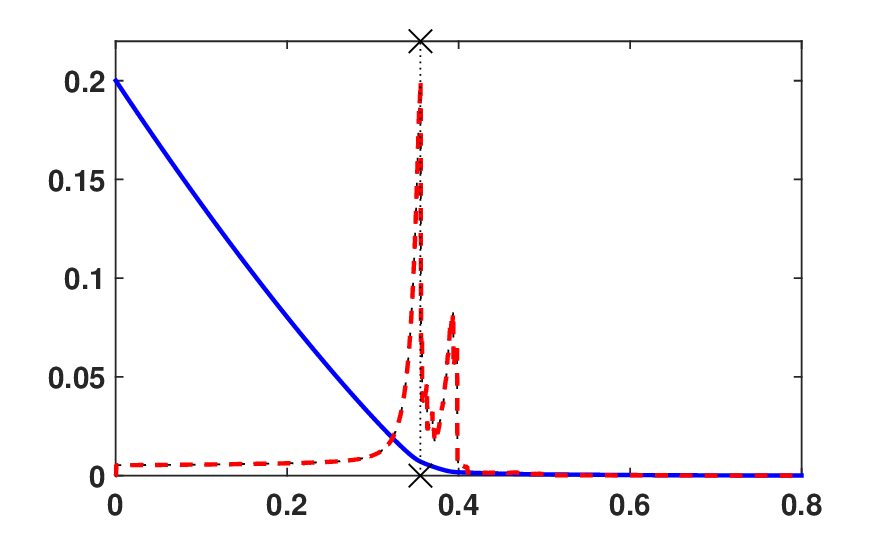}}
\subfloat[]{
\includegraphics[width=1.8in, height=1.6in]{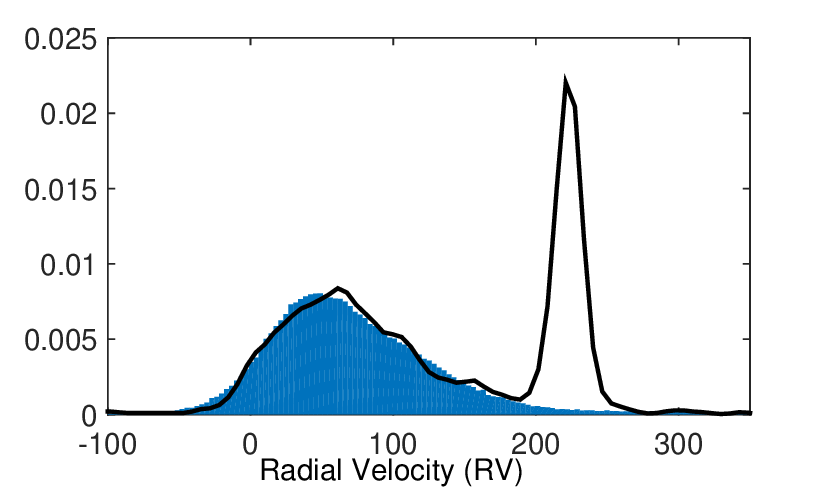}}
\subfloat[]{
\includegraphics[width=1.8in, height=1.6in]{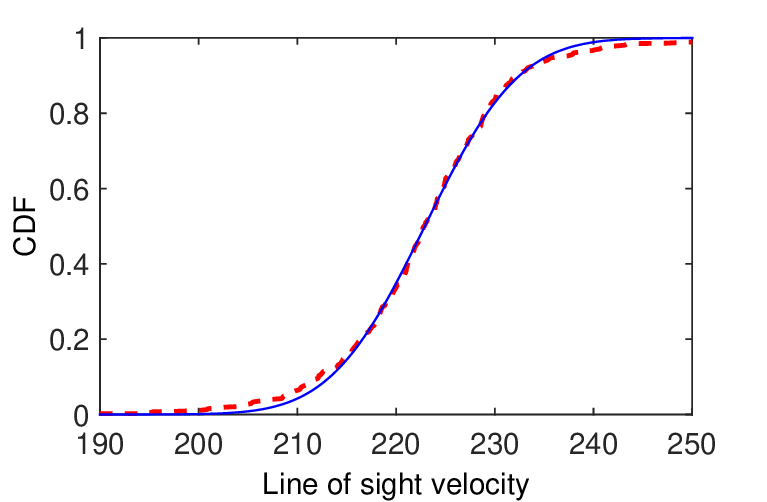}}\\[-5ex]
\caption{Plots for RV data in Carina dSph; left panel: $\gamma d_n(\hat{F}_{s,n}^{\gamma},\check{F}_{s,n}^{\gamma})$ (in solid blue) overlaid with its (scaled) second derivative (in dashed red); middle panel: density of the RV distribution of the contaminating stars overlaid with the (scaled) kernel density estimator of the observed sample; right panel:  $\check{F}_{s,n}^{\tilde{\alpha}_0}$  (in dashed red) overlaid with its closest Gaussian distribution (in solid blue). }%
\label{fig:carinaplot} 

\end{figure}
The 95\% lower confidence bound  for $\alpha_0$ is found to be $0.323$. 
The right panel of Fig.~\ref{fig:carinaplot} shows the estimate of $F_s$ and the closest (in terms of minimising the $L_2(\check{F}_{s,n}^{\tilde{\alpha}_0})$ distance) fitting Gaussian distribution. Astronomers usually assume the distribution of the RVs for these dSph galaxies to be Gaussian. Indeed we see that the estimated $F_s$ is close to a normal distribution (with mean $222.9$ and standard deviation $7.51$), although a formal test of this hypothesis is beyond the scope of the present paper.
The estimate due to \cite{CaiJin10}, $\hat \alpha_0^{CJ}$, is greater than one, while Efron's method (see \cite{Efron07}),  implemented using the ``locfdr" package in R, fails to estimate $\alpha_0.$ 

\section{Concluding remarks}
\label{sec:Conclu}
In this paper we  develop procedures for estimating the mixing proportion and the unknown distribution in a two component mixture model using ideas from shape restricted function estimation. We discuss the identifiability of the model and introduce an identifiable parameter $\alpha_0$, under minimal assumptions on the model. We propose an honest finite sample lower confidence bound of $\alpha_0$ that is distribution-free. Two point estimators of $\alpha_0$, $\hat{\alpha}_0^{c_n}$ and $\tilde \alpha_0$, are studied. We prove that $\hat{\alpha}_0^{c_n}$ is a consistent estimator of $\alpha_0$ and show that the rate of convergence of $\hat{\alpha}_0^{c_n}$ can be arbitrarily close to $\sqrt{n}$, for proper choices of $c_n$. These proposed estimators crucially rely on $\gamma d_n(\hat{F}_{s,n}^{\gamma},\check{F}_{s,n}^{\gamma})$, as a function of $\gamma$, whose plot provides useful insights about the nature of the problem and performance of the estimators.

We observe that the estimators of $\alpha_0$ proposed in this paper have superior finite sample performance than most competing methods. In contrast to most previous work on this topic the  results discussed in this paper hold true even when \eqref{eq:MixMod} is not identifiable. Under the assumption that  \eqref{eq:MixMod} is identifiable, we can find an estimator of $F_s$ which is uniformly consistent. Furthermore, if $F_s$ is known to have a non-increasing density $f_s$ we can find a consistent estimator of $f_s$. All these estimators are tuning parameter free and easily implementable.

We conclude this section by outlining some possible future research directions.  Construction of two-sided confidence intervals for $\alpha_0$ remains a hard problem as the asymptotic distribution of $\hat{\alpha}_0^{c_n}$ depends on the unknown $F$. We are currently developing estimators of $\alpha_0$ when we do not exactly know $F_b$ but  only have an estimator of $F_b$ (e.g., we observe a second i.i.d.~sample from $F_b$). Investigating consistent alternative ways of detecting the ``elbow'' of  the function $\gamma d_n(\hat{F}_{s,n}^{\gamma},\check{F}_{s,n}^{\gamma})$, as an estimator of 
$\tilde{\alpha}_0$, is an interesting future research direction. As we have observed in the astronomy application, formal goodness-of-fit tests for $F_s$ are important  -- they can guide the practitioner to use appropriate parametric models for further analysis --  but are presently unknown. The $p$-values in the prostate data example, considered in Section \ref{sec:prostate}, can have slight dependence. Therefore, investigating the performance and properties of the methods introduced in this paper under appropriate dependence assumptions on $X_1,\ldots, X_n$ is another important direction for future research.

\section*{Acknowledgements}
We thank the Joint Editor, the Associate Editor, and five anonymous referees for their careful reading and constructive comments that lead to an improved version of the paper.
\appendix

\section{Identifiability of $F_s$}\label{sec:Ident_Cont}
In this section we continue the discussion on the identifiability of $F_s$. First, we give some remarks to illustrate Lemmas \ref{lemma:Identifiability for discrete} and \ref{lemma:Identifiability for absolutely continuous}.
\begin{remark}
We consider mixtures of Poisson and binomial distributions to illustrate  Lemma \ref{lemma:Identifiability for discrete}. If $F_s$ is Poisson$(\lambda_s)$ and $F_b$ is Poisson$(\lambda_b)$, then
\begin{equation*}
 \inf_{x\in d(F_b)} \frac{J_{F_s(x)}}{J_{F_b(x)}} = \inf_{k \in \mathbb{N} \cup \{0 \}} {\frac{\lambda_s^k \exp(-\lambda_s)}{\lambda_b^k \exp(-\lambda_b)}}  = \exp(\lambda_b-\lambda_s) \inf_{k \in \mathbb{N}\cup \{0 \}} \left(\frac{\lambda_s}{\lambda_b} \right)^k. 
\end{equation*}
 By an application of Lemma \ref{lemma:Identifiability for discrete}, we have   if $\lambda_s <\lambda_b$ then $\alpha_0=\alpha$; otherwise $\alpha_0=\alpha(1-\exp(\lambda_b-\lambda_s)).$
 
 In the case of a binomial mixture, i.e., $F_s=\text{Bin}(n,p_s)$ and $F_b= \text{Bin}(n,p_b)$,
\begin{equation*}
 \alpha_0= \left\{ \begin{array}{l} \alpha \left[ 1-(\frac{1-p_s}{1-p_b})^n \right],~~~~p_s \ge p_b,~~\\ \alpha \left[1-(\frac{p_s}{p_b})^n \right],~~~~~~p_s < p_b. \end{array} \right.
\end{equation*} 

\end{remark}
\begin{remark}
If $F_s$ is $N(\mu_s, \sigma_s^2)$ and $F_b$ $(\ne F_s)$ is $N(\mu_b,\sigma_b^2)$ then it can be easily shown that the problem is identifiable if and only if $\sigma_s \leq \sigma_b$. When $\sigma_s > \sigma_b,$ the model is not identifiable, an application of Lemma  \ref{lemma:Identifiability for absolutely continuous} gives $\alpha_0= \alpha \big[1- ({\sigma_b}/{\sigma_s})\exp\big(-\sigma_s\sigma_b (\mu_b-\mu_s)^2/2\big)\big]$. Thus,   $\alpha_0$ increases to $\alpha$ as $|\mu_s-\mu_b|$ tends to infinity. It should be noted that the problem is actually identifiable if we  restrict ourselves to the parametric family of a two-component Gaussian mixture model. 
\end{remark}
\begin{remark}
Now consider a mixture of exponential random variables, i.e., $F_s$ is $E(a_s,\sigma_s)$ and $F_b$ $(\ne F_s)$ is $E(a_b,\sigma_b)$, where $E(a,\sigma)$ is the distribution that has the density $(1/\sigma) \exp(-(x-a)/\sigma) \mathbf{1}_{(a,\infty)}(x)$. In this case, the problem is identifiable if $a_s > a_b$, as this implies the support of $F_s$ is a proper subset of the support of $F_b$. But when $a_s \le a_b$, the problem is identifiable if and only if $\sigma_s \leq \sigma_b$.
\end{remark}
\begin{remark}
It is also worth pointing out that even in cases where the problem is not identifiable the difference between the true mixing proportion $\alpha$ and the estimand $\alpha_0$ may be very small. Consider the  hypothesis test $H_0: \theta=0$ versus $H_1: \theta \ne 0$ for the model $N(\theta,1)$ with test statistic $\bar{X}$. The density of the p-values under $\theta$ is $$f_\theta(p) = \frac{1}{2} e^{-m \theta^2/2}[e^{-\sqrt{m}\theta^2 \Phi^{-1}(1-p/2)}+e^{\sqrt{m}\theta^2 \Phi^{-1}(1-p/2)}],$$ where $m$ is the sample size. Here $f_\theta(1)=e^{-m\theta^2/2} >0$, so the model is not identifiable. As $F_b$ is uniform, it can be easily verified that $\alpha_0=\alpha-\alpha \inf_p f_\theta(p)$. However, as the value of $f_\theta$ decreases exponentially with $m$, in many practical situations, where $m$ is not too small, the difference between $\alpha$ and $\alpha_0$ will be negligible. 
\end{remark}

In the following lemma, we try to find the relationship between $\alpha$ and $\alpha_0$ when $F$ is a general CDF.
\begin{lemma}\label{lemma:Identifiabilty_cont}
Suppose that
\be \label{eq:Mixture of Disc and Cont}
F=\kappa F^{(a)}+ (1-\kappa) F^{(d)},
\ee where $F^{(a)}$ is an absolutely continuous CDF and $F^{(d)}$ is a piecewise constant CDF, for some $\kappa \in (0,1)$. Then $$\alpha_0 =\alpha -  \min \left\{ \frac{\alpha \kappa_s -\alpha^{(a)}_0 \kappa}{	\kappa_b} ,  \frac{\alpha (1-\kappa_s)-\alpha^{(d)}_0 ({1-\kappa})}{(1-\kappa_b)}  \right\},$$
 where $\alpha_0^{(a)}$ and $\alpha_0^{(d)}$ are defined as in \eqref{eq:alpha.est}, but with $\{F^{(a)}, F_b^{(a)}\}$ and  $\{F^{(d)}, F_b^{(d)}\}$, respectively (instead of $\{F, F_b\}$). Similarly, $\kappa_s$ and $\kappa_b$ are defined as in \eqref{eq:Mixture of Disc and Cont}, but for $F_s$ and $F_b$, respectively. 
\end{lemma} 

\begin{proof}
From the definition of $\kappa_s$ and $\kappa_b$, we have $ F_s =\kappa_s F_s^{(a)} +(1-\kappa_s) F_s^{(d)},$ and $F_b=\kappa_b F_b^{(a)} +(1-\kappa_b) F_b^{(d)}.$ Thus from \eqref{eq:MixMod}, we get
\begin{equation*}
 F= \alpha \kappa_s F_s^{(a)} +(1-\alpha) \kappa_b F_b^{(a)} +\alpha (1-\kappa_s) F_s^{(d)} + (1-\alpha)(1-\kappa_s) F_b^{(d)}.
\end{equation*}
Now using the definition of $\kappa$, we see that
$\kappa = \alpha \kappa_s+(1-\alpha) \kappa_b,\  1-\kappa = \alpha (1-\kappa_s)+(1-\alpha) (1-\kappa_b).$ 
If we write 
\begin{eqnarray}
F^{(a)}=\alpha^{(a)} F_s^{(a)} +(1-\alpha^{(a)}) F_b^{(a)},\label{eq:submixmodelabs} \nonumber
\end{eqnarray}
it can easily seen that $\alpha^{(a)}= \frac{\alpha \kappa_s}{\kappa}$; and similarly,  $\alpha^{(d)}=\frac{\alpha (1-\kappa_s)}{1-\kappa}$. Then, we can find $\alpha_0^{(d)}$ and $\alpha_0^{(a)}$ as in Lemmas \ref{lemma:Identifiability for discrete} and \ref{lemma:Identifiability for absolutely continuous}, respectively.
Note that  
\begin{eqnarray*}
&&\sup \left\{ 0 \leq \epsilon \leq 1  : \alpha F_s-\epsilon F_b \mbox{  is a sub-CDF} \right\}\\
&=&\sup \left\{ 0 \leq \epsilon \leq 1  : \alpha (\kappa_s F_s^{(a)} +(1-\kappa_s) F_s^{(d)}) -\epsilon (\kappa_b F_b^{(a)} +(1-\kappa_b) F_b^{(d)})\mbox{  is a sub-CDF} \right\}\\
&=&\sup \left\{ 0 \leq \epsilon \leq 1  :  \text{both } \alpha \kappa_s F_s^{(a)}- \epsilon \kappa_b F_b^{(a)}, \alpha (1-\kappa_s) F_s^{(d)}  - \epsilon (1-\kappa_b) F_b^{(d)} \mbox{  are sub-CDFs} \right\}\\
&=& \min \left( \sup \left\{ 0 \leq \epsilon \leq 1  :  \alpha \kappa_s F_s^{(a)}- \epsilon \kappa_b F_b^{(a)} \mbox{ is a sub-CDF}  \right\} \right.,\\ 
	& & \hspace{1in} \left. \sup \left\{ 0 \leq \epsilon \leq 1  :  \alpha (1-\kappa_s) F_s^{(d)} - \epsilon (1-\kappa_b) F_b^{(d)} \mbox{ is a sub-CDF}  \right\} \right )\\ 
&=& \min \left( \frac{\alpha \kappa_s}{	\kappa_b} \essinf \frac{f^{(a)}_s}{f^{(a)}_b}, \frac{\alpha (1-\kappa_s)}{(1-\kappa_b)} \inf_{x \in d(F^{(d)}_b)} \frac{J_{F^{(d)}_s}(x)}{J_{F^{(d)}_b}(x)}       \right)\\
&=&\min \left( \frac{(\alpha \kappa_s -\alpha^{(a)}_0 \kappa)}{	\kappa_b} ,  \frac{(\alpha (1-\kappa_s)-\alpha^{(d)}_0 ({1-\kappa}))}{(1-\kappa_b)}  \right),
\end{eqnarray*}
where $J_{G}$ and $d(J_{G})$ are defined before Lemma \ref{lemma:Identifiability for discrete} and we use the notion that $\frac{0}{0}  =1$. 
Hence, by \eqref{eq:AltDefAlpha} the result follows.
\end{proof}
Lemma \ref{lemma:Identifiabilty} is now a corollary of this result.

\section{Performance comparison of $\hat{\alpha}_0^{c_n},$ $\hat{\alpha}_0^{CV},$  and $\tilde{\alpha}_0$} \label{sec:perfor_cont}
In Figs.~7 and~8 we present further simulation experiments to investigate  the finite sample performance of $\hat{\alpha}_0^{c_n},$ $\hat{\alpha}_0^{CV},$  and $\tilde{\alpha}_0$ across different simulation scenarios. 
In each setting we also include the performance of the best performing competing estimators discussed in Section \ref{sec:PerfEst}.
\begin{figure} \label{fig:cmp_1}
\includegraphics[width=6in,height=5in]{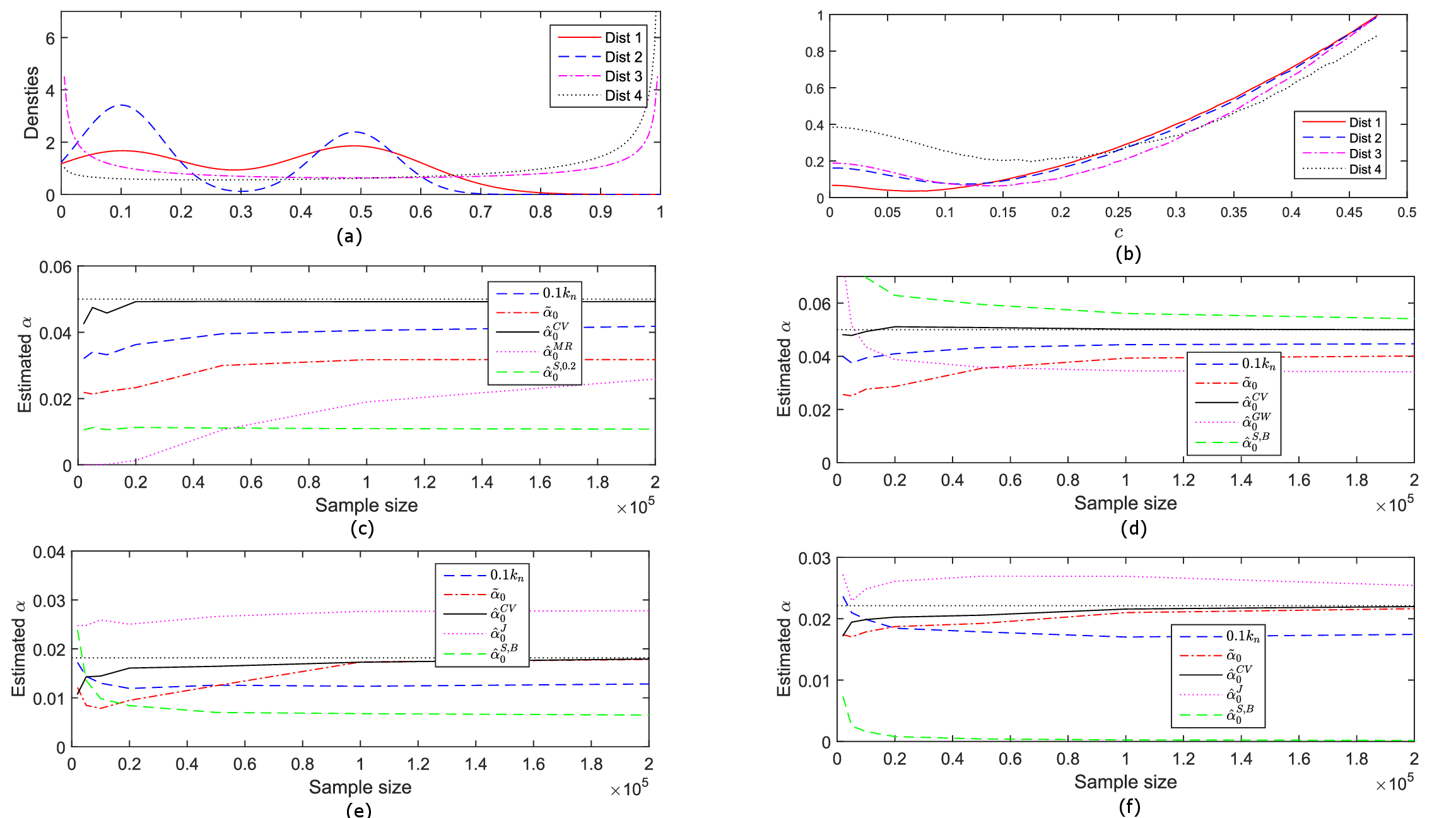}
\caption{Plots comparing the performance of $\hat{\alpha}_0^{c_n},$ $\hat{\alpha}_0^{CV},$  and $\tilde{\alpha}_0$; (a) density functions for four different choices of $F_s$;  (b) plot of the average of $\sum_{k=1}^K \int (\mathbb{F}_n^{k}- \hat{F}^{k})^2 d\mathbb{F}_n^{k}$ (see \eqref{eq:cv} of the main paper), as a function of $c$, computed over 500 independent samples of size 50000  corresponding to Dist 1-4; (c)-(f) gives the means of different competing estimators of $\alpha_0$, computed over 500 independent samples  for Dist 1-4 respectively (in each figure the horizontal dotted black line denotes the true $\alpha_0$).} 
\end{figure}

\begin{figure} \label{fig:cmp_2}
\includegraphics[width=6in,height=5in]{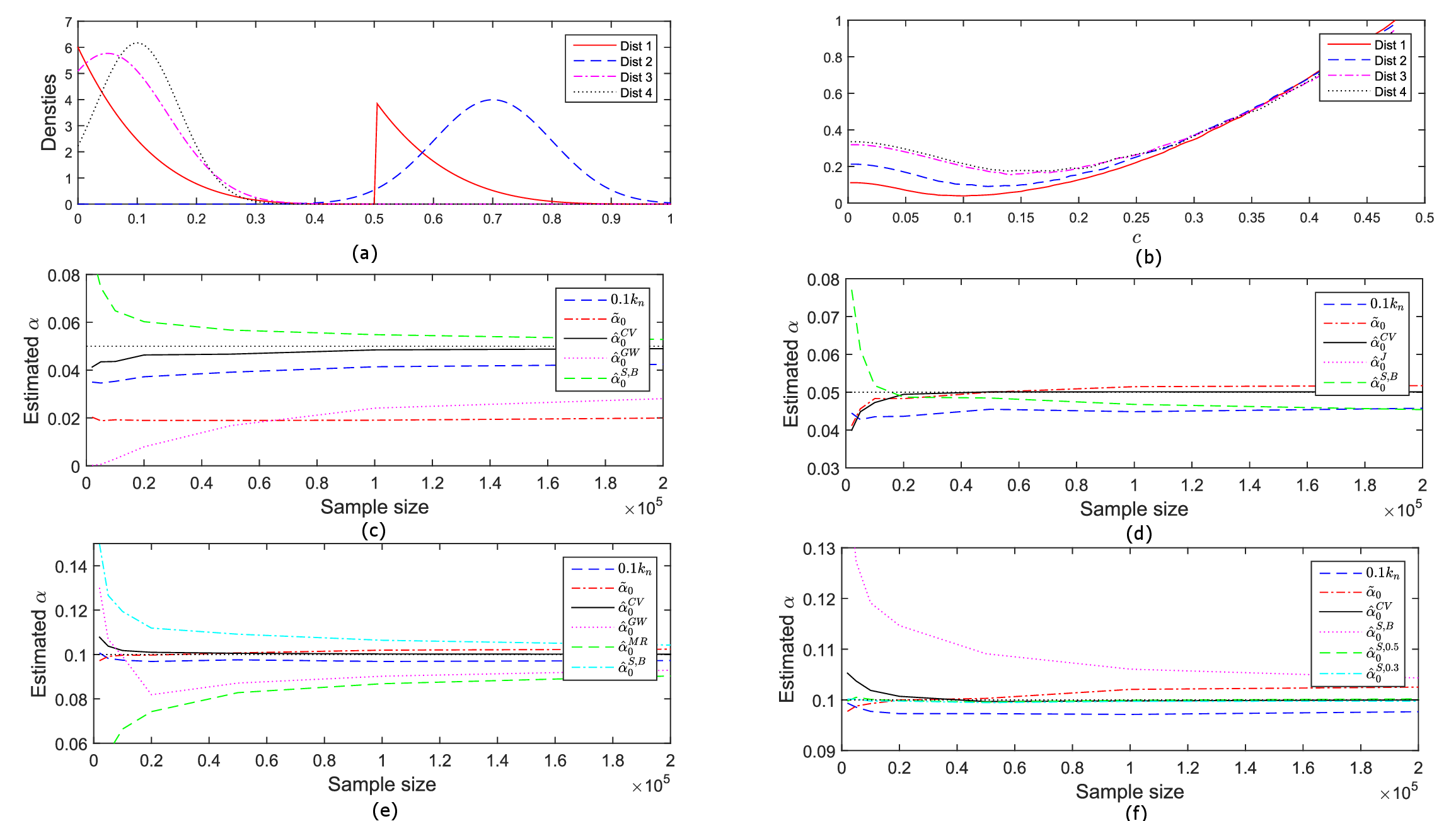}
\caption{Plots comparing the performance of $\hat{\alpha}_0^{c_n},$ $\hat{\alpha}_0^{CV},$  and $\tilde{\alpha}_0$; (a) density functions for four different choices of $F_s$;  (b) plot of the average of $\sum_{k=1}^K \int (\mathbb{F}_n^{k}- \hat{F}^{k})^2 d\mathbb{F}_n^{k}$ (see \eqref{eq:cv} of the main paper), as a function of $c$, computed over 500 independent samples of size 50000  corresponding to Dist 1-4; (c)-(f) gives the means of different competing estimators of $\alpha_0$, computed over 500 independent samples  for Dist 1-4 respectively (in each figure the horizontal dotted black line denotes the true $\alpha_0$).}
\end{figure}

\section{Detection of sparse heterogeneous mixtures} \label{sec:Donhojin}
In this section we draw a connection between the lower confidence bound developed in Section~\ref{sec:Lowrbnd} and the {\it Higher Criticism} method of  \cite{DonohoJin04} for detection of sparse heterogeneous mixtures. The detection of heterogeneity in sparse models arises in many applications, e.g., detection of a disease outbreak (see \cite{KulldorffEtAl05}) or early detection of  bioweapons use (see \cite{DonohoJin04}). Generally, in large scale multiple testing problems, when the non-null effect is sparse it is important to detect the existence of non-null effects (see \cite{CaiEtAl07}).

 \cite{DonohoJin04} consider $n$ i.i.d.~data from one of the two possible situations:
 \begin{eqnarray*}
&&H_0:X_i \sim F_b, \quad 1 \leq i \leq n ,\\
&&H_1^{(n)}:  X_i \sim F^{n}:= \alpha_n F_{n,s}+(1-\alpha_n)F_b,\quad  1 \leq i \leq n,
\end{eqnarray*}
where $\alpha_n \sim n^{-\lambda}$ and $ F_{n,s}$ is such that $d(F_{n,s},F_b)$ is bounded away from 0. In \cite{DonohoJin04} the main focus is on testing $H_0$, i.e., $\alpha_n = 0$. We can test this hypothesis by rejecting $H_0$ when $\hat{\alpha}_L>0$. The following lemma shows that indeed this yields a valid testing procedure for $\lambda < 1/2$.

\begin{thm}\label{lemma:HC}
If $\alpha_n \sim n^{-\lambda}$, for $\lambda < 1/2$, then $P_{H_0}(\mbox{Reject } H_0)  = \beta$  and $P_{H_1^{(n)}}(\hat \alpha_L >0) \rightarrow 1$  as $n \rightarrow \infty$.
\end{thm}
\begin{proof}
Note that $\{\hat \alpha_L>0\}$ is equivalent to $\{c_n \leq \sqrt{n} d_n(\mathbb{F}_n, F_b)\}$ which shows that
\begin{eqnarray*}
c_n & \leq & \sqrt{n} d_n(\mathbb{F}_n, (1-\alpha_n)F_b+ \alpha_n F_{n,s})+ \sqrt{n} d_n(\alpha_n F_b, \alpha_n F_{n,s})\\
& = & \sqrt{n} d_n(\mathbb{F}_n, F^{n}) + \alpha_n \sqrt{n} d_n( F_{n,s},F_b),
\end{eqnarray*}
where $c_n$ is chosen as in Theorem \ref{thm:BoundAlpha}. It is easy to see that  $\sqrt{n} d_n(\mathbb{F}_n, F^{n})$ is $O_P(1)$ and $\alpha_n \sqrt{n} d_n(F_{n,s},F_b) \rightarrow\infty$, for $\lambda <1/2$, which shows that $P_{H_1^{(n)}}(\hat{\alpha}_L >0) \rightarrow 1$.  It can be easily seen that $P_{H_0}(\hat{\alpha}_L >0)= P_{H_0}(\mbox{Reject } H_0)  = \beta$.
\end{proof}

\section{Proofs of theorems and lemmas in the main paper} \label{sec:proofs}

\subsection{Proof of Lemma \ref{lemma:Non-Identifiability}}
From the definition of $\alpha_0$, we have
\begin{eqnarray*}
\alpha_0 &=& \inf \left\{ 0 \leq \gamma \leq \alpha  : [F-(1-\gamma)F_b]/{\gamma} \mbox{  is a valid CDF} \right\} \notag \\
&=& \inf \left\{ 0 \leq \gamma \leq \alpha  : [\alpha F_s + (1-\alpha) F_b-(1-\gamma)F_b]/{\gamma} \mbox{  is a valid CDF} \right\} \notag \\
&=& \inf \left\{ 0 \leq \gamma \leq \alpha  : [\alpha F_s -(\alpha-\gamma)F_b]/{\gamma} \mbox{  is a valid CDF} \right\} \notag \\
&=& \alpha -\sup \left\{ 0 \leq \epsilon \leq \alpha  : \alpha F_s-\epsilon F_b \mbox{  is a sub-CDF} \right\}\notag \\
&=& \alpha -\sup \left\{ 0 \leq \epsilon \leq 1  : \alpha F_s-\epsilon F_b \mbox{  is a sub-CDF} \right\},
\end{eqnarray*}
where  the final equality follows from the fact that if $\epsilon> \alpha$, then $\alpha F_s- \epsilon F_b$ will not be a sub-CDF.  

To show that $\alpha_0 = 0$ if and only if $F= F_b$ let us define $\delta=\alpha- \epsilon$. 
Note that $\alpha_0=0$, if and only if 
\baa 
&&\sup \left\{ 0 \leq \epsilon \leq 1  : \alpha F_s-\epsilon F_b \mbox{  is a sub-CDF} \right\}=\alpha\\
&\Leftrightarrow& \inf \left\{ 0 \leq \delta \leq 1  : \alpha( F_s- F_b) +\delta F_b \mbox{  is a sub-CDF} \right\}=0.
\eaa
However, it is easy to see that the last equality is true if and only if $F_s-F_b \equiv 0.$ 

\subsection{Proof of Lemma \ref{lemma:Identifiability for discrete}}

When $d(F_b)\not \subset d(F_s)$, there exists a $x\in d(F_b) - d(F_s)$, i.e., there exists a $x$ which satisfies $F_b(x)-F_b(x-) >0$ and $F_s(x)-F_s(x-)=0$. Then for all $\epsilon>0$, $F_s(x-) - \epsilon F_b(x-) > F_s(x)- \epsilon F_b(x)$. This shows that $F_s - \epsilon F_b$ cannot be a sub-CDF, and hence by Lemma~\ref{lemma:Non-Identifiability} the model is identifiable. Now let us assume that $d(F_b)\subset d(F_s)$. 
\begin{eqnarray*}
\left\{ 0 \leq \epsilon \leq 1  : \alpha F_s-\epsilon F_b \mbox{  is a sub-CDF} \right\} &=&\left\{ 0 \leq \epsilon \leq 1 : \alpha J_{F_s}(x)- \epsilon J_{F_b}(x) \geq 0, ~\forall x\in d(J_{F_b})\right \}\\
&=& \left\{ 0 \leq \epsilon \leq 1 :  \frac{ J_{F_s}(x)} {J_{F_b}(x)} \geq \frac{\epsilon}{\alpha},~\forall x\in d(J_{F_b})\right \}\\
&=& \left\{ 0 \leq \epsilon \leq 1 :\inf_{x\in d({F_b})} \frac{ J_{F_s}(x)} {J_{F_b}(x)} \geq \frac{\epsilon}{\alpha}\right \}.
\end{eqnarray*}
Therefore, using \eqref{eq:AltDefAlpha}, we get the desired result.  

\subsection{Proof of Lemma \ref{lemma:Identifiability for absolutely continuous}}
From \eqref{eq:AltDefAlpha}, we have 
\begin{eqnarray*}
\alpha_0 &=&  \alpha -\sup \left\{ 0 \leq \epsilon \leq 1  : \alpha F_s-\epsilon F_b \mbox{  is a sub-CDF} \right\}\\
&=&  \alpha - \sup \left\{ 0 \leq \epsilon \leq 1  : \alpha f_s(x) - \epsilon f_b(x) \ge 0  \mbox{ almost every } x \right\}\\
&=& \alpha - \sup \left\{ 0 \leq \epsilon \leq 1  : \alpha \frac{f_s}{f_b}(x) \ge \epsilon    
\mbox{ almost every } x \right\}\\
&=& \alpha \left\{1- \essinf \frac{f_s}{f_b}\right\}.
\end{eqnarray*}

\subsection{Proof of Theorem \ref{thm:Distribution_free}}
Without loss of generality, we can assume that $F_b$ is the uniform distribution on $(0,1)$ and, for clarity, in the following we write $U$ instead of $F_b$. Let us define 
\baa
 A&:=&\left\{ \gamma \in (0,1]: \frac{F-(1-\gamma)U}{\gamma} \mbox{  is a valid CDF} \right\},\\
A^Y &: =&\left\{ \gamma \in (0,1]: \frac{G-(1-\gamma) U\circ \Psi}{\gamma} \mbox{  is a valid CDF} \right\}.
 \eaa
Since  $\alpha_0 = \inf  A,$ and $\alpha_0^Y = \inf A^Y$ for  the first part of the theorem it is enough to show that $A= A^Y.$ Let us first show that $A^Y \subset A$. Suppose $\eta \in A^Y$.  We first show that ${(F - (1-\eta) U)}/{\eta}$ is a non-decreasing function. For all $t_1 \leq t_2$, we have that 
\bee
\frac{G(t_1)-(1-\eta) U(\Psi(t_1))}{\eta} \leq \frac{G(t_2)-(1-\eta) U(\Psi(t_2))}{\eta}.
\eee
Let  $ y_1\leq y_2$. Then,
\bee
 \frac{G(\Psi^{-1}(y_1))-(1-\eta) U(\Psi(\Psi^{-1}(y_1)))}{\eta} \leq \frac{G(\Psi^{-1}(y_2))-(1-\eta) U(\Psi(\Psi^{-1}(y_2)))}{\eta}, 
\eee
since $y_1\leq y_2\Rightarrow \Psi^{-1}(y_1)\leq \Psi^{-1}(y_2)$.
However, as $\Psi$ is continuous, $\Psi(\Psi^{-1}(y))=y$ and $G(\Psi^{-1}(y))= \alpha F_s (y) + (1-\alpha) U(y) =F(y)$. Hence, we have 
$$ \frac{F(y_1)-(1-\eta)U(y_1)}{\eta} \leq \frac{F(y_2)-(1-\eta)U(y_2)}{\eta}. $$
As $F$ and $U$ are CDFs, it is easy to see that  $\lim_{x \to -\infty} {(F(x)-(1-\eta)U(x))}/{\eta}=0$, \\ $\lim_{x \to \infty} {(F(x)-(1-\eta)U(x))}/{\eta}=1$ and ${(F-(1-\eta)U)}/{\eta}$ is a right continuous function. Hence, for $\eta \in A^Y$, ${(F-(1-\eta)U)}/{\eta}$ is a CDF and thus, $\eta \in A$. We can similarly prove $A \subset A^Y$. Therefore, $A=A^Y$ and $\alpha_0=\alpha_0^Y$.
%
%
%

Note that $$\gamma   d_n(\hat{F}_{s,n}^{\gamma},\check{F}_{s,n}^\gamma) = \min_{W \in \mathcal{F}}\frac{1}{n} \sum_{i=1}^n \{W(X_i) - \hat{F}_{s,n}^{\gamma}(X_i)\}^2,$$ where  $\mathcal{F}$ is the class of all CDFs. For the second part of theorem it is enough to show  that 
$$\min_{W \in \mathcal{F}}\frac{1}{n} \sum_{i=1}^n \{W(X_i) - \hat{F}_{s,n}^{\gamma}(X_i)\}^2=\min_{B \in \mathcal{F}}\frac{1}{n} \sum_{i=1}^n \{B(Y_i) - \hat{G}_{s,n}^{\gamma}(Y_i)\}^2.$$
First note that
\baa
\mathbb{G}_n(y) &=& \frac{1}{n} \sum_{i=1}^n \mathbf{1}\{ \Psi^{-1} (X_i)\le y\}\\
 &=& \frac{1}{n} \sum_{i=1}^n \mathbf{1}\{  X_i\le\Psi(y)\} \\
 &=& \mathbb{F}_n(\Psi (y)).
 \eaa
 Thus, from the definition of $\hat{G}_{s,n}^{\gamma},$ we have 
  \baa
\hat{G}_{s,n}^{\gamma}(Y_i)&=&\frac{\mathbb{F}_n(\Psi (Y_i)) -(1-\gamma) U(\Psi(Y_i))}{\gamma}\\
&=& \frac{\mathbb{F}_n(X_i) -(1-\gamma)U(X_i)}{\gamma} = \hat{F}_{s,n}^{\gamma} (X_i).
\eaa
Therefore,
\baa
\frac{1}{n} \sum_{i=1}^n \{B(Y_i) - \hat{G}_{s,n}^{\gamma}(Y_i)\}^2&=&\frac{1}{n} \sum_{i=1}^n \{B(Y_i) - \hat{F}_{s,n}^{\gamma} (X_i)\}^2\\
&=& \frac{1}{n} \sum_{i=1}^n \{B(\Psi^{-1}(X_i)) - \hat{F}_{s,n}^{\gamma} (X_i)\}^2 \\
&=& \frac{1}{n} \sum_{i=1}^n \{W(X_i) - \hat{F}_{s,n}^{\gamma} (X_i)\}^2,
\eaa
where $W(x) := B(\Psi^{-1}(x) )$. $W$ is a valid CDF as $\Psi^{-1}$ is non-decreasing.

\subsection{Proof of Lemma \ref{lemma1}}
Letting ${F}_{s}^{\gamma} = (F - (1-\gamma)F_b)/\gamma$,
observe that $$\gamma d_n(\hat{F}_{s,n}^\gamma,F_s^\gamma)=d_n(F,\mathbb{F}_n).$$ Also note that ${F}_{s}^{\gamma}$ is a valid CDF for $\gamma \ge \alpha_0$.
As $\check{F}_{s,n}^\gamma $ is defined as the function that minimises the $L_2(\mathbb{F}_n)$ distance of $\hat{F}_{s,n}^\gamma$ over all CDFs,
$$\gamma d_n(\check{F}_{s,n}^\gamma,\hat{F}_{s,n}^\gamma) \leq \gamma d_n(\hat{F}_{s,n}^\gamma, F^\gamma_s) = d_n(F,\mathbb{F}_n).$$

To prove the second part of the lemma, notice that for $\gamma \geq \alpha_0$ the result follows from above and the fact that $d_n(F,\mathbb{F}_n) \stackrel{a.s.}{\rightarrow} 0$ as $n \rightarrow \infty$.
		
For $\gamma < \alpha_0$, $F_s^\gamma$ is not a valid CDF, by the definition of $\alpha_0$. Note that as $n \rightarrow \infty$, $\hat{F}^{\gamma}_{s,n} \stackrel{a.s.}{\rightarrow} F^{\gamma}_s$ point-wise. So, for large enough $n$, $\hat{F}^{\gamma}_{s,n}$ is not a valid CDF, whereas $\check{F}_{s,n}^\gamma$ is always a CDF. Thus, $d_n(\hat{F}_{s,n}^{\gamma},\check{F}_{s,n}^\gamma)$ converges to something positive.

\subsection{Proof of Lemma \ref{lemma2}}
Assume that $\gamma_1\leq\gamma_2$ and $\gamma_1,\gamma_2 \in A_n$. If $\gamma_3=\eta \gamma_1+(1-\eta)\gamma_2$, for $0\leq \eta \leq 1$, it is easy to observe from \eqref{eq:naive} that  
\be
\eta (\gamma_1 \hat{F}^{\gamma_1}_{s,n}) +(1-\eta)(\gamma_2\hat{F}^{\gamma_2}_{s,n}) =\gamma_3\hat{F}^{\gamma_3}_{s,n}. \nonumber
\ee
Note that $[\eta (\gamma_1 \check{F}^{\gamma_1}_{s,n}) +(1-\eta)(\gamma_2\check{F}^{\gamma_2}_{s,n})]/\gamma_3$ is a valid CDF, and thus from the definition of  $\check{F}^{\gamma_3}_{s,n}$, we have
\begin{eqnarray} \label{eq:ConvexCrit}
d_n(\hat{F}^{\gamma_3}_{s,n}, \check{F}^{\gamma_3}_{s,n} ) & \leq & d_n \left( \hat{F}^{\gamma_3}_{s,n}, [\eta (\gamma_1 \check{F}^{\gamma_1}_{s,n}) +(1-\eta)(\gamma_2\check{F}^{\gamma_2}_{s,n})]/{\gamma_3} \right) \nonumber \\
& = &  d_n \Big( \frac{\eta (\gamma_1 \hat{F}^{\gamma_1}_{s,n}) +(1-\eta)(\gamma_2\hat{F}^{\gamma_2}_{s,n})}{\gamma_3}, \frac{\eta (\gamma_1 \check{F}^{\gamma_1}_{s,n}) +(1-\eta)(\gamma_2\check{F}^{\gamma_2}_{s,n}) }{\gamma_3} \Big) \nonumber \\
& \leq &  \frac{\eta \gamma_1}{\gamma_3} d_n(\hat{F}^{\gamma_1}_{s,n}, \check{F}^{\gamma_1}_{s,n} ) + \frac{(1-\eta) \gamma_2}{\gamma_3} d_n(\hat{F}^{\gamma_2}_{s,n}, \check{F}^{\gamma_2}_{s,n} ) 
\end{eqnarray}
where the last step follows from the triangle inequality. But as $\gamma_1, \gamma_2 \in A_n$, the above inequality yields
\begin{equation*}
d_n(\hat{F}^{\gamma_3}_{s,n}, \check{F}^{\gamma_3}_{s,n} ) \leq \frac{\eta \gamma_1}{\gamma_3}\frac{c_n}{\sqrt{n}\gamma_1} + \frac{(1 - \eta) \gamma_2}{\gamma_3}\frac{c_n}{\sqrt{n}\gamma_2}  =\frac{c_n} {\sqrt{n}\gamma_3}.
\end{equation*}
Thus $\gamma_3 \in A_n$.
\subsection{Proof of Theorem \ref{thm:ConsAlpha}}
We need to show that $P(\left|  \hat{\alpha}_0^{c_n}- \alpha_0 \right| > \epsilon ) \rightarrow 0 $ for any $\epsilon >0$. Let us first show that $$P(  \hat{\alpha}_0^{c_n}- \alpha_0  < -\epsilon ) \rightarrow 0. $$ The statement is obviously true if $\alpha_0 \le \epsilon$. So let us assume that $\alpha_0 > \epsilon$. Suppose $\hat{\alpha}_0^{c_n} -\alpha_0 < -\epsilon$, i.e., $\hat{\alpha}_0^{c_n}  < \alpha_0 -\epsilon$. Then by the definition of $\hat{\alpha}_0^{c_n}$ and the convexity of $A_n$, we have $(\alpha_0-\epsilon) \in A_n$ (as $A_n$ is a convex set in $[0,1]$ with $1 \in A_n$ and $\hat{\alpha}_0^{c_n}  \in A_n$), and thus
\be d_n(\hat{F}_{s,n}^{\alpha_0-\epsilon},\check{F}_{s,n}^{\alpha_0-\epsilon}) \leq \frac{c_n }{\sqrt{n}(\alpha_0 - \epsilon)}.
\label{proof1}
\ee But by \eqref{eq:claim3} the left-hand side of \eqref{proof1} goes to a non-zero constant in probability. Hence, if $c_n/\sqrt{n}  \rightarrow 0$, $$ P(  \hat{\alpha}_0^{c_n}- \alpha_0  < -\epsilon ) \le P \left( d_n(\hat{F}_{s,n}^{\alpha_0-\epsilon},\check{F}_{s,n}^{\alpha_0-\epsilon} ) \leq \frac{c_n}{\sqrt{n} (\alpha_0-\epsilon)} \right) \rightarrow 0.$$ This completes the proof of the first part of the claim.

Now suppose that $\hat{\alpha}_0^{c_n}- \alpha_0 >\epsilon.$ Then,
\begin{eqnarray*}
\hat{\alpha}_0^{c_n}- \alpha_0  >\epsilon  &\Rightarrow& \sqrt{n}d_n(\hat{F}_{s,n}^{\alpha_0+\epsilon },\check{F}_{s,n}^{\alpha_0 +\epsilon}) \geq  \frac{c_n}{\alpha_0+\epsilon}  \\
&\Rightarrow& \sqrt{n} d_n( \mathbb{F}_n,F) \geq c_n.
\end{eqnarray*}
The first implication follows from the definition of $\hat{\alpha}_0^{c_n}$, while the second implication is true by Lemma \ref{lemma1}. The right-hand side of the last inequality is (asymptotically similar to) the Cram\'{e}r--von Mises statistic for which the asymptotic distribution is well-known and thus if $c_n \rightarrow \infty$ the result follows.
\subsection{Proof of Lemma \ref{lemma:LimDistnull}}
As $\alpha_0=0$, 
\be
P(\hat{\alpha}_0^{c_n} =0) = 1- P(\hat{\alpha}_0^{c_n} > 0)= 1- P(\sqrt{n} d_n(\F_n,F) > c_n ) \rightarrow 1,
\ee
since $\sqrt{n} d_n(\mathbb{F}_n,F) = O_P(1)$ by Theorem \ref{thm:ConvDist}.

\subsection{Proof of Theorem~\ref{lemma:RateLeftSide}}\label{sec:proof_lemma:RateLeftSide}
As the proof of this result is slightly involved we break it into a number of lemmas (whose proofs are provided later in this sub-section) and give the main arguments below.

We need to show that given any $\epsilon>0$, we can find an $M > 0$ and $n_0 \in \mathbb{N}$ (depending on $\epsilon$) for which $\sup_{n >n_0} P(r_n|\hat{\alpha}_0^{c_n}-\alpha_0| > M) \le \epsilon.$
\begin{lemma} \label{lemma:RateRightSide}
If $c_n\rightarrow \infty$, then for any $M >0$, $\sup_{n >n_0} P\left(r_n(\hat{\alpha}_0^{c_n}-\alpha_0)>M\right) < \epsilon$, for large enough $n_0 \in \mathbb{N}$.
\end{lemma}

Finding an $r_n$ such that $P\left(r_n(\hat{\alpha}_0^{c_n}-\alpha_0)<- M\right) < \epsilon$ for large enough $n$ is more complicated.  We start with some notation. Let $\mathcal{F}$ be the class of all CDFs and $\mathbb{H}$ be the Hilbert space $L_2(F) := \{ f: \R \to \R | \int f^2 dF < \infty\}$.  For a closed convex subset $\mathcal{K}$ of $\mathbb{H}$ and $ h \in \mathbb{H}$, we define the projection of $h$ onto $\mathcal{K}$ as 
\be \label{eq:DefProj}
\Pi(h|\mathcal{K}) := \argmin_{f \in \mathcal{K}} d(f,h),   
\ee
where $d$ stands for the $L_2(F)$ distance, i.e., if $g,h \in \mathbb{H}$, then $d^2(g,h) = \int (g - h)^2 dF. $ We define the tangent cone of $\mathcal{F}$ at $f_0 \in \mathcal{F}$, as 
\be \label{eq:DefTancone}
T_\mathcal{F}(f_0) := { \left\{ \lambda(f - f_0): \lambda \ge 0, f \in \mathcal{F}\right\}}.
\ee
For any $H \in \mathcal{F}$ and $\gamma >0$, let us define
\begin{equation*}
\hat H^{\gamma} :=  \frac{H - (1 - \gamma) F_b}{\gamma}, \quad 	\check{H}_n^{\gamma}  :=  \argmin_{G \in \mathcal{F}} \gamma d_n(\hat H^{\gamma}, G),\quad \text{and}\quad 	\bar{H}_n^{\gamma}  :=  \argmin_{G \in \mathcal{F}} \gamma d(\hat H^{\gamma}, G).
\end{equation*} 
For $H = \mathbb{F}_n$ and $\gamma = \alpha_0$ we define the  three quantities above and call them $\hat F_{s,n}^{\alpha_0}$, $\check{F}_{s,n}^{\alpha_0}$, and $\bar{F}_{s,n}^{\alpha_0}$ respectively. Note that  \begin{equation}\label{eq:Step1}
  P \left(r_n(\hat{\alpha}_0^{c_n}-\alpha_0)<- M \right) = P(\sqrt{n} \gamma_n\ d_n(\hat F_{s,n}^{\gamma_n}, \check{F}_{s,n}^{\gamma_n}) <c_n),
\end{equation} 
where $\gamma_n= \alpha_0 - M/r_n$. To study the limiting behavior of $d_n(\hat F_{s,n}^{\gamma_n}, \check{F}_{s,n}^{\gamma_n})$ we break it as the sum of $ d_n(\hat F_{s,n}^{\gamma_n}, \check{F}_{s,n}^{\gamma_n}) - d(\hat F_{s,n}^{\gamma_n}, \bar F_{s,n}^{\gamma_n})$ and $d(\hat F_{s,n}^{\gamma_n}, \bar F_{s,n}^{\gamma_n})$. The following two lemmas (proved in Sections~\ref{sec:proof_emma:Diff1} and \ref{sec:proof_lemma:Diffpart2} respectivley) give the asymptotic behavior of the two terms. The proof of Lemma~\ref{lemma:Diffpart2} uses the functional delta method (cf.~Theorem 20.8 of \cite{vanderVaart98}) for the projection operator; see Theorem 1 of  \cite{Fils-VilletardEtAl08}. 

\begin{lemma}\label{lemma:Diff1}
If $\sqrt{n}/r_n^2 \rightarrow 0,$ then $
U_n := \sqrt{n} \gamma_n d_n(\hat F_{s,n}^{\gamma_n}, \check{F}_{s,n}^{\gamma_n}) - \sqrt{n} \gamma_n d(\hat F_{s,n}^{\gamma_n}, \bar F_{s,n}^{\gamma_n}) \stackrel{P}{\rightarrow} 0.$
\end{lemma}

\begin{lemma} \label{lemma:Diffpart2}
If $c_n \rightarrow \infty$, then
\bee  \label{eq:Diffpart2}
\frac{\sqrt{n}\gamma_n}{c_n M} d(\hat F_{s,n}^{\gamma_n}, \bar F_{s,n}^{\gamma_n}) \stackrel{P}{\rightarrow} \left\{\int V^2 dF  \right\}^{1/2} > 0 
\eee
where $$V := (F_s^{\alpha_0}-F_b) -\Pi \left( \left. F_s^{\alpha_0}-F_b \right\vert T_\mathcal{F}(F_s^{\alpha_0})\right) \ne 0$$ and
\begin{equation}\label{eq:F_s_alp_0}
 F_s^{\alpha_0} := \frac{F-(1-\alpha_0)F_b}{\alpha_0}.
\end{equation}
\end{lemma} 
Using \eqref{eq:Step1}, and the notation introduced in the above two lemmas we see that 
\begin{eqnarray}
P \left(r_n(\hat{\alpha}_0^{c_n}-\alpha_0)<-M  \right) &=& P\left( \frac{1}{c_n} U_n+ \frac{\sqrt{n}\gamma_n}{c_n} d(\hat F_{s,n}^{\gamma_n}, \bar F_{s,n}^{\gamma_n}) <1\right). \label{eq:ratebound2}
\end{eqnarray}
However, $U_n \stackrel{P}{\rightarrow} 0$ (by Lemma \ref{lemma:Diff1}) and $\frac{\sqrt{n}\gamma_n}{c_n M} d(\hat F_{s,n}^{\gamma_n}, \bar F_{s,n}^{\gamma_n}) \stackrel{P}{\rightarrow} \int V^2 dF$ (by Lemma \ref{lemma:Diffpart2}).  The result now follows from \eqref{eq:ratebound2}, by taking a large enough $M.$

\subsubsection{Proof of Lemma \ref{lemma:RateRightSide}} \label{sec:proof_lemma:RateRightSide}
Note that
\begin{eqnarray*}
P ( r_n(\hat{\alpha}_0^{c_n}-\alpha_0)>M) \leq P\left(\hat{\alpha}_0^{c_n}>\alpha_0 \right)&=&P \left( \sqrt{n}\alpha_0  d_n(\hat{F}_{s,n}^{\alpha_0},\check{F}_{s,n}^{\alpha_0})>c_n  \right)\\
&\leq& P\left(\sqrt{n} \alpha_0 d_n(\hat{F}_{s,n}^{\alpha_0}, F_s^{\alpha_0}) >c_n\right)\\
&=&P\left( \sqrt{n} d_n(\mathbb{F}_n,F) >c_n \right) \rightarrow 0,
\end{eqnarray*}
as $c_n\rightarrow \infty$, since $\sqrt{n} d_n(\mathbb{F}_n,F) = O_P(1)$. Therefore, the result holds for sufficiently large $n$.

\subsubsection{Proof of Lemma \ref{lemma:Diff1}}\label{sec:proof_emma:Diff1} 
It is enough to show that
\begin{eqnarray}
W_n := n \gamma_n^2 d_n^2(\hat F_{s,n}^{\gamma_n}, \check{F}_{s,n}^{\gamma_n}) - n \gamma_n^2 d^2(\hat F_{s,n}^{\gamma_n}, \bar F_{s,n}^{\gamma_n}) & \stackrel{P}{\rightarrow}& 0, \label{eq:Diff2}
\end{eqnarray}
since $U_n^2 \le |W_n|$. Note that 
\begin{eqnarray*}
\check{F}_{s,n}^{\gamma_n} & = & \argmin_{G \in \mathcal{F}} d_n(\mathbb{F}_n, \gamma_n G + (1 - \gamma_n) F_b),\\
\bar{F}_{s,n}^{\gamma_n} & = & \argmin_{G \in \mathcal{F}} d(\mathbb{F}_n, \gamma_n G + (1 - \gamma_n) F_b).
\end{eqnarray*}
For each positive integer $n$ and $c >0$, we introduce the following classes of functions:
\begin{eqnarray*}
\mathcal{G}_c(n) = \left\{ \sqrt{n}(G - (1 - \gamma_n) F_b - \gamma_n \check G_{n}^{\gamma_n})^2 : G \in \mathcal{F}, \  \|G - F\| <\frac{c}{\sqrt{n}}\right\}, \\
\mathcal{H}_c(n) = \left\{ \sqrt{n}(H - (1 - \gamma_n) F_b - \gamma_n \bar H_{n}^{\gamma_n})^2 : H \in \mathcal{F}, \  \|H - F\| <\frac{c}{\sqrt{n}}\right\}.
\end{eqnarray*}
Let us also define $$D_n :=\sup_{t \in \mathbb{R}} \sqrt{n}|\F_n(t)-F(t)| = \|\F_n-F\|.$$ From the definition of the minimisers $\check{F}_{s,n}^{\gamma_n}$ and $\bar{F}_{s,n}^{\gamma_n}$, we see that
\begin{eqnarray}\label{eq:BoundMin}
\gamma_n^2 \ |d_n^2(\hat F_{s,n}^{\gamma_n}, \check{F}_{s,n}^{\gamma_n}) -  d^2(\hat F_{s,n}^{\gamma_n}, \bar F_{s,n}^{\gamma_n})| \le \max \left \{ |(d_n^2 - d^2)(\mathbb{F}_n, \gamma_n \check{F}_{s,n}^{\gamma_n} + (1 - \gamma_n) F_b) | \right., \nonumber \\
\left.|(d_n^2 - d^2)(\mathbb{F}_n, \gamma_n \bar{F}_{s,n}^{\gamma_n} + (1 - \gamma_n) F_b) | \right\}.\;\;\;\;\;\;\;\;\;
\end{eqnarray}
Observe that
\begin{eqnarray*}
n \gamma_n^2 \ [(d_n^2 - d^2)(\mathbb{F}_n, \gamma_n \check{F}_{s,n}^{\gamma_n} + (1 - \gamma_n) F_b)] = \sqrt{n} (\mathbb{P}_n - P)[g_n] = \nu_n(g_n),
\end{eqnarray*}
where $g_n := \sqrt{n} \{ \F_n - \gamma_n \check{F}_{s,n}^{\gamma_n} - (1 - \gamma_n) F_b\}^2$, $\mathbb{P}_n$ denotes the empirical measure of the data, and $\nu_n := \sqrt{n} (\mathbb{P}_n - P)$ denotes the usual empirical process. Similarly, $$n \gamma_n^2 \ [(d_n^2 - d^2)(\mathbb{F}_n, \gamma_n \bar{F}_{s,n}^{\gamma_n} + (1 - \gamma_n) F_b)] = \sqrt{n} (\mathbb{P}_n - P)[h_n] = \nu_n(h_n),$$
where $h_n :=\sqrt{n} \{ \F_n - \gamma_n \bar{F}_{s,n}^{\gamma_n} - (1 - \gamma_n) F_b\}^2$. Thus, combining \eqref{eq:Diff2}, \eqref{eq:BoundMin} and the above two displays, we get, for any $\delta >0$,
\begin{eqnarray}
 P(|W_n| > \delta) \le P \left( |\nu_n(g_n)| > \delta \right) +  P \left( |\nu_n(h_n)| > \delta \right). \label{eq:EPbound}
\end{eqnarray}
The first term in the right hand side of \eqref{eq:EPbound} can be bounded above as
\begin{eqnarray}
P (|\nu_n({g}_n)|>\delta) & = & P (|\nu_n({g}_n)|>\delta, {g}_n \in \mathcal{G}_c(n))+ P(|\nu_n({g}_n)|>\delta, {g}_n \notin \mathcal{G}_c(n)) \nonumber \\
&\leq & P(|\nu_n({g}_n)|>\delta, {g}_n \in \mathcal{G}_c(n)) + P({g}_n \notin \mathcal{G}_c(n)) \nonumber \\
&\leq & P \left(\sup_{g \in \mathcal{G}_c(n)}|\nu_n({g})|>\delta \right) + P({g}_n \notin \mathcal{G}_c(n)) \nonumber \\
& \leq & \frac{1}{\delta}E \left( \sup_{g \in \mathcal{G}_c(n)}|\nu_n({g})| \right)+ P({g}_n \notin \mathcal{G}_c(n)) \nonumber \\
&\leq & J_{[\;]} \frac{P[G^2_{c,n}]}{\delta} + P({g}_n \notin \mathcal{G}_c(n)), \label{eq:Boundv_n0}
\end{eqnarray}
where $G_{c,n} :=6 c^2/\sqrt{n} + 16\sqrt{n}   \frac{M^2}{r_n^2 } \|F_s^{\alpha_0} -F_b\|^2$ is an envelope for $\mathcal{G}_c(n)$ and $J_{[\;]}$ is a constant. Note that to derive the last inequality, we have used the maximal inequality in Corollary (4.3) of \cite{Pollard89}; the class $\mathcal{G}_c(n)$ is ``manageable'' in the sense of \cite{Pollard89} (as a consequence of equation (2.5) of \cite{VandeGeer00}).

To see that $G_{c,n}$ is an envelope for $\mathcal{G}_c(n),$ observe that for any $G \in \mathcal{F}$, 
\begin{eqnarray*}
G - (1 - \gamma_n) F_b &=& G - F + \frac{M}{r_n} (F_s^{\alpha_0} -F_b) +\gamma_n F_s^{\alpha_0}.
\end{eqnarray*}
Hence,
 $$ F_s^{\alpha_0}- \frac{M}{r_n \gamma_n} \|F_s^{\alpha_0} -F_b\|  -\frac{\|G - F\|}{\gamma_n} \le  \frac{G - (1 - \gamma_n) F_b}{\gamma_n} \le F_s^{\alpha_0}+ \frac{M}{r_n \gamma_n}\|F_s^{\alpha_0} -F_b\|  +\frac{\|G - F\|}{\gamma_n}.$$
As the two bounds are monotone, from the  properties of isotonic estimators (see e.g., Theorem 1.3.4 of \cite{RWD88}), we can always find a version of $\check{G}_s^{\gamma_n}$ such that 
$$ F_s^{\alpha_0}- \frac{M}{r_n \gamma_n} \|F_s^{\alpha_0} -F_b\|  -\frac{\|G - F\|}{\gamma_n} \le  \check{G}_s^{\gamma_n} \le F_s^{\alpha_0}+ \frac{M}{r_n \gamma_n}\|F_s^{\alpha_0} -F_b\|  +\frac{\|G - F\|}{\gamma_n}.$$
Therefore,
\begin{equation}\label{eq:MontoneBnd}
-2 \frac{M}{r_n} \|F_s^{\alpha_0} -F_b\|  -\|G - F\| \le  \gamma_n \check{G}_s^{\gamma_n}- \gamma_n F_s^{\alpha_0} -\frac{M}{r_n }(F_s^{\alpha_0} -F_b)  \le  2\frac{M}{r_n }\|F_s^{\alpha_0} -F_b\|  +\|G - F\|.
\end{equation} 
Thus, for $\sqrt{n}(G - (1 - \gamma_n) F_b - \gamma_n \check G_{s}^{\gamma_n})^2  \in \mathcal{G}_c(n)$, 
\begin{eqnarray*}
(G - (1 - \gamma_n) F_b - \gamma_n \check G_{s}^{\gamma_n})^2 &=& \left[ (G - F)  + \left( \gamma_n \check G_{s}^{\gamma_n} -\gamma_n F_s^{\alpha_0} -  \frac{M}{r_n} (F_b-F_s^{\alpha_0}) \right) \right]^2 \\
& \le & 2  (G - F)^2 + 2 \left( \gamma_n \check G_{s}^{\gamma_n} -\gamma_n F_s^{\alpha_0} -  \frac{M}{r_n} (F_b-F_s^{\alpha_0})\right) ^2 \\
& \le & 2 \| G - F \|^2 + 2 \left( 2\frac{M}{r_n }\| F_s^{\alpha_0} -F_b\|  +\|G - F\| \right) ^2 \\
& \le & 6  \| G - F \|^2 + 16 \frac{M^2}{r_n^2 } \|F_s^{\alpha_0} -F_b\|^2 \\
& \le & 6 c^2 + 16 \frac{M^2}{r_n^2 } \|F_s^{\alpha_0} -F_b\|^2 = \frac{G_{c,n}}{\sqrt{n}},
\end{eqnarray*}
where the second inequality follows from \eqref{eq:MontoneBnd}.
From the definition of ${g}_n$ and $D_n^2$, we have  $|{g}_n(t)| \leq \frac{6}{\sqrt{n}} D_n^2 +16\sqrt{n}   \frac{M^2}{r_n^2 } \|F_s^{\alpha_0} -F_b\|^2$, for all $t \in \mathbb{R}$. As $D_n  = O_P(1)$, for any given $\epsilon>0$, there exists $c>0$ (depending on $\epsilon$) such that
\begin{eqnarray}
\label{eq:BoundG_c0}
P({g}_n \notin \mathcal{G}_c(n)) &=& P\left( \|\F_n -F \|\geq \frac{c}{\sqrt{n}} \right) = P(D_n\geq c) \leq \epsilon,
\end{eqnarray}
for all sufficiently large $n$. 

Therefore, for any given $\delta >0$ and $\epsilon >0$, we can make both $J \{6 \frac{c^2}{\sqrt{n}} + 16\sqrt{n}   \frac{M^2}{r_n^2 } \|F_s^{\alpha_0} -F_b\|^2\}^2$ and $P({g}_n \notin \mathcal{G}_c(n))$ less than $\epsilon$ for large enough $n$ and $c(>0)$, using the fact that $\sqrt{n}/r_n^2 \rightarrow 0$ and (\ref{eq:BoundG_c0}). Thus, $P(|\nu_n({g}_n)|>\delta) \leq 2\epsilon$ by \eqref{eq:Boundv_n0}.

A similar analysis can be done for the second term of \eqref{eq:EPbound}. The result now follows.

\subsubsection{Proof of Lemma \ref{lemma:Diffpart2}}\label{sec:proof_lemma:Diffpart2}
Note that
\begin{equation} \label{eq:splitDiff2}
\frac{\sqrt{n}\gamma_n}{c_n} (\hat F_{s,n}^{\gamma_n}- \bar F_{s,n}^{\gamma_n}) = \frac{\sqrt{n}\gamma_n}{c_n}  (\hat F_{s,n}^{\gamma_n}- F^{\alpha_0}_s) -\frac{\sqrt{n}\gamma_n}{c_n} (\bar F_{s,n}^{\gamma_n}- F^{\alpha_0}_s). \nonumber
\end{equation}
However, a simplification yields
\begin{eqnarray*}
\frac{\sqrt{n}\gamma_n}{c_n}  (\hat F_{s,n}^{\gamma_n}- F^{\alpha_0}_s) 
=\frac{1}{c_n}  \sqrt{n}(\mathbb{F}_n -F)  +\frac{\sqrt{n} M}{c_n r_n \alpha_0}(F- F_b ).
\end{eqnarray*}
Since $\sqrt{n}(\mathbb{F}_n -F)/{c_n}$ is $o_P(1)$, $\sqrt{n}=c_n r_n$, and $F - F_b = \alpha_0 (F^{\alpha_0}_s - F_b)$,  we have
\begin{equation} \label{eq:NaiveConv}
\frac{\sqrt{n}\gamma_n}{c_n M}  (\hat F_{s,n}^{\gamma_n}- F^{\alpha_0}_s)  \stackrel{P}{\rightarrow}  F_s^{\alpha_0}-F_b \quad \mbox{ in } \mathbb{H}.
\end{equation}
By applying the functional delta method (see Theorem 20.8 of \cite{vanderVaart98}) for the projection operator (see Theorem 1 of  \cite{Fils-VilletardEtAl08}) to \eqref{eq:NaiveConv}, we have 
\begin{equation}\label{eq:ProjNaiveConv}
\frac{\sqrt{n}\gamma_n}{c_n M}  (\bar{ F}_{s,n}^{\gamma_n}- F^{\alpha_0}_s)   \stackrel{P}{\rightarrow} \Pi \left( \left. F_s^{\alpha_0} - F_b \right\vert T_\mathcal{F}(F_s^{\alpha_0})\right) \quad \mbox{ in } \mathbb{H}.
\end{equation}
By combining \eqref{eq:NaiveConv} and \eqref{eq:ProjNaiveConv}, we have
\begin{equation}\label{eq:ProjNaiveSum}
\frac{\sqrt{n}\gamma_n}{c_n M} (\hat F_{s,n}^{\gamma_n}- \bar F_{s,n}^{\gamma_n})  \stackrel{P}{\rightarrow}  (F_s^{\alpha_0} - F_b) -\Pi \left( \left. F_s^{\alpha_0} - F_b \right\vert T_\mathcal{F}(F_s^{\alpha_0})\right) \quad \mbox{ in } \mathbb{H}.
\end{equation}
The result now follows by applying the continuous mapping theorem to \eqref{eq:ProjNaiveSum}.
We prove $V\ne 0$ by contradiction. Suppose that $V=0$, i.e., $(F_s^{\alpha_0} - F_b) \in T_\mathcal{F}(F_s^{\alpha_0})$. Therefore, for some distribution function $G$ and $\eta >0$, we have $V=(\eta+1) F_s^{\alpha_0} -  F_b - \eta G,$ by the definition of $T_\mathcal{F}(F_s^{\alpha_0})$. By the discussion leading to \eqref{eq:AltDefAlpha}, it can be easily seen that $\eta G$ is a sub-CDF, while $(\eta+1) F_s^{\alpha_0}-F_b$ is not (as that would contradict \eqref{eq:AltDefAlpha}). Therefore, $V \ne 0$ and thus $\int V^2 dF >0$.

\subsection{Proof of Theorem \ref{lemma:LimDist}} \label{sec:proof_lemma:LimDist}
The constant $c$ defined in the statement of the theorem can be explicitly expressed as $$c =-\left\{ \int V^2 dF \right\}^{-\frac{1}{2}} ,$$ where $$V= (F_s-F_b)- \Pi(F_s-F_b| T_\mathcal{F} (F_s)),$$ and $\Pi$ and $T_\mathcal{F} (\cdot)$ are defined in~\eqref{eq:DefProj} and~\eqref{eq:DefTancone}, respectively.

Let $x>0$. Obviously, 
\bee
P(r_n (\hat{\alpha}_0^{c_n} -\alpha_0) \le x) = 1-P(r_n (\hat{\alpha}_0^{c_n} -\alpha_0) > x).
\eee
By Lemma \ref{lemma:RateRightSide}, we have that  $P(r_n (\hat{\alpha}_0^{c_n} -\alpha_0) > x) \rightarrow 0$ if $c_n \rightarrow \infty$. Now let $x \le 0$. In this case the left hand side of the above display equals $P( \sqrt{n} \gamma_n d_n(\hat{F}^{\gamma_n}_{s,n}, \check{F}^{\gamma_n}_{s,n})  \le c_n),$
where $ \gamma_n= \alpha_0 + x/r_n$.  A simplification yields
\be
\frac{\sqrt{n}}{c_n}\gamma_n  (\hat F_{s,n}^{\gamma_n}-F^{\alpha_0}_s) 
 \stackrel{P}{\rightarrow}  -x(F_s^{\alpha_0}-F_b), \mbox{ in } \mathbb{H},  \label{eq:splitDiff3Part1and2} 
\ee
since$\sqrt{n}(\mathbb{F}_n -F)/{c_n}$ is $o_P(1)$; see the proof of Lemma \ref{lemma:Diffpart2} (Section \ref{sec:proof_lemma:Diffpart2}) for the details. By applying the functional delta method (cf.~Theorem 20.8 of \cite{vanderVaart98}) for the projection operator (see Theorem 1 of  \cite{Fils-VilletardEtAl08}) to \eqref{eq:splitDiff3Part1and2}, we have
\be
\frac{\sqrt{n}}{c_n}\gamma_n (\bar F_{s,n}^{\gamma_n}- F^{\alpha_0}_s) \stackrel{d}{\rightarrow} \Pi\left(\left.-x(F_s^{\alpha_0}-F_b)\right\vert T_\mathcal{F} (F_s^{\alpha_0}) \right) \quad \mbox{ in } \mathbb{H}. \label{eq:splitDiff3Part3}
\ee
Adding \eqref{eq:splitDiff3Part1and2} and \eqref{eq:splitDiff3Part3}, we get
\begin{equation*}
\frac{\sqrt{n}}{c_n}\gamma_n (\hat F_{s,n}^{\gamma_n}- \bar F_{s,n}^{\gamma_n}) \rightarrow -x(F_s^{\alpha_0}-F_b)- \Pi\left(\left.-x(F_s^{\alpha_0}-F_b)\right\vert T_\mathcal{F} (F_s^{\alpha_0}) \right) \quad \mbox{ in } \mathbb{H}.
\end{equation*}
By the continuous mapping theorem, we get $\sqrt{n}/c_n \gamma_n d(\hat F_{s,n}^{\gamma_n}, \bar F_{s,n}^{\gamma_n}) \stackrel{P}{\rightarrow} |x| \left\{\int V^2 dF\right\}^{1/2}.$
Hence, by Lemma \ref{lemma:Diff1}, 
\begin{equation*}
P(r_n (\hat{\alpha}_0^{c_n} -\alpha_0) \le x) \rightarrow
\begin{dcases}
    1 ,& \text{if } x > 0,\\
    1, & \text{if $x \le0$ and }  |x| \leq \left\{\int V^2 dF\right\}^{-1/2},\\
    0,              & \text{otherwise}.
\end{dcases}
\end{equation*}

\subsection{Proof of Theorem \ref{thm:BoundAlpha}}
Letting $c_n =H_n^{-1} (1-\beta),$ we have
\begin{eqnarray*}
	P(\alpha_0  \ge \hat \alpha_L) & = & P \left(\sqrt{n} \alpha_0 \ d_n(\hat{F}_{s,n}^{\alpha_0} ,\check{F}_{s,n}^{\alpha_0}) \le  c_n \right) \\
	& \ge& P \left( \sqrt{n} \alpha_0 \ d_n(\hat{F}_{s,n}^{\alpha_0} ,{F}_{s}^{\alpha_0}) \le  c_n \right)  =   H_n(c_n) = 1-\beta,
\end{eqnarray*}
where we have used the fact that $\alpha_0 d_n(\hat{F}_{s,n}^{\alpha_0},{F}_{s}^{\alpha_0})= d_n(\mathbb{F}_{n},F)$. Note that, when $\alpha_0=0$, $F=F_b$, and using \eqref{eq:DistNull} we get $$P(\alpha_0  \ge \hat \alpha_L)= P \left(\sqrt{n}  \ d_n(\F_n,F_b) \le  c_n \right) =P \left(\sqrt{n}  \ d_n(\F_n,F) \le  c_n \right)=1-\beta.$$
\subsection{Proof of Theorem \ref{thm:ConvDist}}
It is enough to show that $\sup_x |H_n(x) - G(x)| \stackrel{}{\rightarrow} 0$, where $G$ is the limiting distribution of the Cram\'{e}r-von Mises statistic, a continuous distribution. As $\sup_x |G_n(x) - G(x)| \stackrel{}{\rightarrow} 0$,  it is enough to show that
\begin{equation}
\sqrt{n} d_n(\mathbb{F}_n,F)-\sqrt{n} d(\mathbb{F}_n,F) \stackrel{P}{\rightarrow} 0.
\label{eq:lemma4.1}
\end{equation}
We now prove \eqref{eq:lemma4.1}. Observe that
\begin{eqnarray}\label{eq:ConvProb}
n (d_n^2 - d^2) (\mathbb{F}_n,F) = \sqrt{n} (\mathbb{P}_n - P)[\hat{g}_n] = \nu_n(\hat g_n),
\end{eqnarray}
where $\hat{g}_n =\sqrt{n}( \F_n-F)^2$, $\mathbb{P}_n$ denotes the empirical measure of the data, and $\nu_n := \sqrt{n} (\mathbb{P}_n - P)$ denotes the usual empirical process. We will show that $\nu_n(\hat g_n) \stackrel{P}{\rightarrow} 0$, which will prove \eqref{eq:ConvProb}.

For each positive integer $n$, we introduce the following class of functions $$\mathcal{G}_c(n) = \left\{ \sqrt{n}(H-F)^2 : H \in \mathcal{F} \mbox{ and }\sup_{t \in \mathbb{R}}|H(t)-F(t)|<\frac{c}{\sqrt{n}}\right\}.$$
Let us also define $$D_n :=\sup_{t \in \mathbb{R}} \sqrt{n}|\F_n(t)-F(t)|.$$ From the definition of $\hat{g}_n$ and $D_n^2$, we have  $\hat{g}_n(t) \leq \frac{1}{\sqrt{n}} D_n^2$, for all $t \in \mathbb{R}$. As $D_n  = O_P(1)$, for any given $\epsilon>0$, there exists $c>0$ (depending on $\epsilon$) such that
\begin{equation}\label{eq:BoundG_c}
P( \hat{g}_n \notin \mathcal{G}_c(n)) = P(\sqrt{n}\sup_t|\hat{g}_n(t)| \geq c^2) = P(D_n^2 \geq c^2)\leq \epsilon,
\end{equation}
for all sufficiently large $n$. Therefore, for any $\delta >0$, using the same sequence of steps as in \eqref{eq:Boundv_n0},
\begin{eqnarray}
P(|\nu_n(\hat{g}_n)|>\delta) 
&\leq & J_{[\;]} \frac{E[G^2_c(n)]}{\delta} + P(\hat{g}_n \notin \mathcal{G}_c(n)), \label{eq:Boundv_n}
\end{eqnarray}
where $G_c(n) := \frac{c^2}{\sqrt{n}}$ is an envelope for $\mathcal{G}_c(n)$ and $J_{[\;]}$ is a constant. Note that to derive the last inequality we have used the maximal inequality in Corollary (4.3) of Pollard (1989); the class $\mathcal{G}_c(n)$ is ``manageable'' in the sense of \cite{Pollard89} (as a consequence of equation (2.5) of \cite{VandeGeer00}).

Therefore, for any given $\delta >0$ and $\epsilon >0$, for large enough $n$ and $c>0$ we can make both $J_{[\;]}c^4/(\delta n)$ and $P(\hat{g}_n \notin \mathcal{G}_c(n))$ less than $\epsilon$, using \eqref{eq:BoundG_c} and \eqref{eq:Boundv_n}, and thus, $P(|\nu_n(\hat{g}_n)|>\delta) \leq 2\epsilon$. The result now follows.

\subsection{Proof of Theorem \ref{lemma:LwrBndAsym}} \label{sec:LwrBndAsym_proof}
The random variable $U$ defined in the statement of the theorem can be explicitly expressed as 
$$U:=\bigg[\int \left\{\mathbb{G}_F- \Pi( \mathbb{G}_F| T_\mathcal{F}(F_s^{\alpha_0})\right\}^2  dF\bigg]^{1/2},$$ where $\mathbb{G}_F$ is the $F$-Brownian bridge.

By the same line of arguments as in the proof of Lemma \ref{lemma:Diff1} (see Section~\ref{sec:proof_emma:Diff1}), it can be easily seen that $ \sqrt{n} \alpha_0\  d_n(\hat F_{s,n}^{\alpha_0}, \check{F}_{s,n}^{\alpha_0}) - \sqrt{n} \alpha_0 \ d(\hat F_{s,n}^{\alpha_0}, \bar F_{s,n}^{\alpha_0}) \stackrel{P}{\rightarrow} 0.$ Moreover, by Donsker's theorem,  $$ \sqrt{n} \alpha_0 (\hat F_{s,n}^{\alpha_0}-  F_s^{\alpha_0}) \stackrel{d}{\rightarrow} \mathbb{G}_F.$$ By applying the functional delta method for the projection operator, in conjunction with the continuous mapping theorem to the previous display, we have $$ \sqrt{n} \alpha_0 (\bar{F}_{s,n}^{\alpha_0}-  F_s^{\alpha_0}) \stackrel{d}{\rightarrow} \Pi (\mathbb{G}_F| T_\mathcal{F}(F_s^{\alpha_0})) \quad \text{in} \quad  \mathbb{H},$$  where $\Pi$, $T_\mathcal{F} (\cdot),$ and $F_s^{\alpha_0}$  are defined in~\eqref{eq:DefProj}, \eqref{eq:DefTancone}, and~\eqref{eq:F_s_alp_0}, respectively.  Hence, by an application of the continuous mapping theorem, we have $\sqrt{n} \alpha_0 d(\hat F_{s,n}^{\alpha_0}, \bar F_{s,n}^{\alpha_0}) \stackrel{d}{\rightarrow} U$. The result now follows.
\subsection{Proof of Lemma \ref{lemma:DecCritFn}}
Let $0 < \gamma_1 < \gamma_2 < 1$. Then,
\begin{eqnarray*}
\gamma_2  d_n(\hat{F}_{s,n}^{\gamma_2},\check{F}_{s,n}^{\gamma_2}) & \le & \gamma_2  d_n(\hat{F}_{s,n}^{\gamma_2},(\gamma_1/ \gamma_2) \check{F}_{s,n}^{\gamma_1} + (1 - \gamma_1/ \gamma_2) {F}_{b}) \\
& = & d_n(\gamma_1 \hat{F}_{s,n}^{\gamma_1} + (\gamma_2 - \gamma_1) F_b, \gamma_1 \check{F}_{s,n}^{\gamma_1} + (\gamma_2 - \gamma_1) {F}_{b}) \\
& \le & \gamma_1  d_n(\hat{F}_{s,n}^{\gamma_1},\check{F}_{s,n}^{\gamma_1}),
\end{eqnarray*}
which shows that $\gamma d_n(\hat{F}_{s,n}^{\gamma},\check{F}_{s,n}^\gamma) $ is a non-increasing function. To show that $\gamma d_n(\hat{F}_{s,n}^{\gamma},\check{F}_{s,n}^\gamma) $ is convex, let $0 < \gamma_1 < \gamma_2 < 1$ and $\gamma_3 = \eta \gamma_1 + (1 -\eta) \gamma_2$, for $0 \le \eta \le 1$. Then, by \eqref{eq:ConvexCrit} we have the desired result.

\subsection{Proof of Theorem \ref{thm:ConsF_sn}}
The constant $c$ and the function $Q$  defined in the statement of the theorem can be explicitly expressed  as $$c =d(Q,\Pi \left(Q| T_\mathcal{F}(F_s)\right)),$$ and  $$Q:=  (F_s-F_b) \left\{\alpha_0^2\int V^2 dF\right\}^ {-1/2},$$ where $$r_n= \sqrt{n}/c_n, \;\; \;V= (F_s-F_b)- \Pi(F_s-F_b| T_\mathcal{F} (F_s)),$$ and $\Pi$ and $T_\mathcal{F} (\cdot)$ are defined in~\eqref{eq:DefProj} and~\eqref{eq:DefTancone}, respectively. 

Recall the notation of Section~\ref{sec:proof_lemma:RateLeftSide}. Note that from (\ref{eq:naive}), $$\hat F_{s,n}^{\check \alpha_n}(x) = \frac{\alpha_0}{\check \alpha_n} F_s(x) + \frac{\check \alpha_n - \alpha_0}{\check \alpha_n} F_b(x) + \frac{(\F_n - F)(x)}{\check \alpha_n},$$
for all $x \in \R$. Thus we can bound $\hat F_{s,n}^{\check \alpha_n}(x)$ as follows:
\begin{eqnarray*}
\frac{\alpha_0}{\check \alpha_n} F_s(x) - \frac{|\check \alpha_n - \alpha_0|}{\check \alpha_n} - \frac{D'_n}{\check \alpha_n} \le \hat F_{s,n}^{\check \alpha_n}(x) \le \frac{\alpha_0}{\check \alpha_n} F_s(x) + \frac{|\check \alpha_n - \alpha_0|}{\check \alpha_n} + \frac{D'_n}{\check \alpha_n},
\end{eqnarray*}
where $D'_n = \sup_{x \in \R} |\F_n(x) - F(x)|.$  As both the upper and lower bounds are monotone, we can always find a version of $\check F_{s,n}^{\check \alpha_n}$ such that 
\begin{eqnarray*}
\frac{\alpha_0}{\check \alpha_n} F_s- \frac{|\check \alpha_n - \alpha_0|}{\check \alpha_n} - \frac{D'_n}{\check \alpha_n} \le \check F_{s,n}^{\check \alpha_n} \le \frac{\alpha_0}{\check \alpha_n} F_s+ \frac{|\check \alpha_n - \alpha_0|}{\check \alpha_n} + \frac{D'_n}{\check \alpha_n}.
\end{eqnarray*}
Therefore, \begin{eqnarray*}
|\check F_{s,n}^{\check \alpha_n} - F_s| & \le & \frac{|\alpha_0 - \check \alpha_n|}{\check \alpha_n} F_s + \frac{|\check \alpha_n - \alpha_0|}{\check \alpha_n} + \frac{D'_n}{\check \alpha_n} \\
& \le & 2 \frac{|\alpha_0 - \check \alpha_n|}{\check \alpha_n} + \frac{D'_n}{\check \alpha_n} \stackrel{P}{\rightarrow} 0,
\end{eqnarray*}
as $n \rightarrow \infty$, using the fact $\check \alpha_n \stackrel{P}{\rightarrow} \alpha_0 \in (0,1)$. 
Furthermore, if $q_n(\check{\alpha}_n-\alpha_0) =O_P(1)$, where $q_n/ \sqrt{n}\rightarrow0$, it is easy to see that $q_n |\check F_{s,n}^{\check \alpha_n} - F_s|=O_P(1)$, as $q_n  D'_n =o_P(1)$.
Note that 
\[
r_n \hat{\alpha}_0^{c_n} (\hat{F}_{s,n}^{\hat{\alpha}_0^{c_n}} -F_s) = r_n(\mathbb{F}_n-F) + r_n (\alpha_0-\hat{\alpha}_0^{c_n}) \ (F_s-F_b)
\]
Thus\[ 
 \sup_{x \in \mathbb{R}}    | r_n (\hat{F}_{s,n}^{\hat{\alpha}_0^{c_n}} -F_s)(x)  - Q(x)| \stackrel{P}{\rightarrow} 0.\]
 
 Hence by an application of  functional delta method for the projection operator, in conjunction with the continuous mapping theorem, we have
\[
 r_n d(\check{F}_{s,n}^{\hat{\alpha}_0^{c_n}} ,F_s)\stackrel{P}{\rightarrow} d(Q,\Pi(Q| T_\mathcal{F}(F_s))).  \]

\subsection{Proof of Theorem \ref{thm:Consden_sn}}
Let $\epsilon_n := \sup_{x \in \R} | \check F_{s,n}^{\check \alpha_n}(x) - F_s(x) |$. Then the function $F_s + \epsilon_n$ is concave on $[0,\infty)$ and majorises $\check F_{s,n}^{\check \alpha_n}$. Hence, for all $x \in [0,\infty)$, $\check F_{s,n}^{\check \alpha_n}(x) \le F_{s,n}^{\dagger}(x) \le F_s(x) + \epsilon_n$, as $F_{s,n}^{\dagger}$ is the LCM of $\check F_{s,n}^{\check \alpha_n}$. Thus,
$$ -\epsilon_n \le \check F_{s,n}^{\check \alpha_n}(x) - F_s(x) \le F_{s,n}^{\dagger}(x) - F_s(x) \le \epsilon_n,$$ and therefore, $$ \sup_{x \in \R} |F_{s,n}^{\dagger}(x) - F_s(x)| \le \epsilon_n.$$ By Theorem \ref{thm:ConsF_sn}, as $\epsilon_n \stackrel{P}{\rightarrow} 0$, we must also have (\ref{eq:SupConvF_snDagger}).

The second part of the result follows immediately from the lemma is page 330 of \cite{RWD88}, and is similar to the result in Theorem 7.2.2 of that book.

\bibliographystyle{chicago}
\bibliography{SigNoise}
\end{document}